\makeatletter
\@ifundefined{reserveinserts}{%
  \newcommand\reserveinserts[1]{}%
}{}
\makeatother
\makeatletter
\providecommand{\@history@dates}{}
\makeatother

\makeatletter
\AtBeginEnvironment{algorithmic}{\renewcommand{\ttfamily}{\rmfamily}}
\makeatother

\documentclass[APA,Times2COL]{WileyNJDv5}


\makeatletter
\long\def\@makecaption#1#2{%
  \vskip3pt
  \begingroup
  \normalsize
  \sbox\@tempboxa{{\bfseries #1}\quad #2}%
  \ifdim\wd\@tempboxa>\hsize
    {\bfseries #1}\quad #2\par
  \else
    \global\@minipagefalse
    \hb@xt@\hsize{\box\@tempboxa\hfil}%
  \fi
  \endgroup
  \vskip4pt
}
\makeatother

\articletype{Research Article}%

\journal{Computer-Aided Civil and Infrastructure Engineering}
\volume{00}
\copyyear{2026}
\startpage{1}

\usepackage[table]{xcolor}   
\usepackage{float}
\definecolor{headergray}{gray}{0.80}   
\definecolor{rowgray}{gray}{0.93}      
\usepackage{subcaption}

\usepackage{pifont}

\newcommand{\grayabstract}[1]{%
  \begingroup
  \setlength{\fboxsep}{6pt}
  \noindent
  \colorbox{gray!15}{%
    \parbox{0.66\linewidth}{#1}
  }%
  \endgroup
}

\begin{document}

\title{{Energy-Aware Reinforcement Learning for Robotic Manipulation of Articulated Components in Infrastructure Operation and Maintenance}}

\author[1,2]{Xiaowen Tao}

\author[3]{Yinuo Wang}

\author[1]{Haitao Ding}

\author[4]{Yuanyang Qi}

\author[1]{Ziyu Song}

\authormark{TAO \textsc{et al.}}
\titlemark{Energy-Aware Reinforcement Learning for Robotic Manipulation of Articulated Components in Infrastructure Operation and Maintenance}

\address[1]{\orgdiv{National Key Laboratory of Automotive Chassis Integration and Bionics}, \orgname{Jilin University}, \orgaddress{\state{Jilin}, \country{China}}}

\address[2]{\orgdiv{School of Computer Science and Statistics}, \orgname{Trinity College Dublin}, \orgaddress{\state{Dublin}, \country{Ireland}}}

\address[3]{\orgdiv{Arene}, \orgname{Woven by Toyota}, \orgaddress{\state{California}, \country{USA}}}

\address[4]{\orgdiv{Department of Civil Engineering}, \orgname{The University of Hong Kong}, \orgaddress{\state{Pokfulam}, \country{Hong Kong}}}

\corres{Xiaowen Tao, School of Computer Science and Statistics, Trinity College Dublin, Dublin, D02 PN40, Ireland. 
\email{taox@tcd.ie} \\
Ziyu Song, National Key Laboratory of Automotive Chassis Integration and Bionics, Jilin University, Changchun, 130000, China.
\email{songziyu@jlu.edu.cn}}

\fundingInfo{Joint Funds of the National Natural Science Foundation of China under Grant U1864206.}


\abstract[]{%
  \grayabstract{%
  \normalsize
    \textbf{Abstract}\par

With the growth of intelligent civil infrastructure and smart cities, operation and maintenance (O\&M) increasingly requires safe, efficient, and energy-conscious robotic manipulation of articulated components, including access doors, service drawers, and pipeline valves. However, existing robotic approaches either focus primarily on grasping or target object-specific articulated manipulation, and they rarely incorporate explicit actuation energy into multi-objective optimisation, which limits their scalability and suitability for long-term deployment in real O\&M settings. Therefore, this paper proposes an articulation-agnostic and energy-aware reinforcement learning framework for robotic manipulation in intelligent infrastructure O\&M. The method combines part-guided 3D perception, weighted point sampling, and PointNet-based encoding to obtain a compact geometric representation that generalises across heterogeneous articulated objects. Manipulation is formulated as a Constrained Markov Decision Process (CMDP), in which actuation energy is explicitly modelled and regulated via a Lagrangian-based constrained Soft Actor-Critic scheme. The policy is trained end-to-end under this CMDP formulation, enabling effective articulated-object operation while satisfying a long-horizon energy budget. Experiments on representative O\&M tasks demonstrate 16\%-30\% reductions in energy consumption, 16\%-32\% fewer steps to success, and consistently high success rates, indicating a scalable and sustainable solution for infrastructure O\&M manipulation. A repository is hosted at \url{https://github.com/allen-legged-robot/csac-arm-rl}.

  }%
}

\maketitle

\section{Introduction}\label{arch}
With the rapid development of intelligent civil infrastructure and smart cities, operation and maintenance (O\&M) has become a critical component in ensuring the safety, functionality, and sustainability of large-scale infrastructure systems. However, current infrastructure O\&M practices still rely heavily on manual intervention, which is often associated with high operational costs, significant safety risks, and low efficiency \cite{sanchez2016maintenance}. 

Some O\&M tasks are frequently performed in confined, hazardous, or hard-to-reach environments, posing considerable challenges for human operators and increasing the likelihood of operational errors. Moreover, infrastructure O\&M typically involves long-term and continuous operation, where energy consumption and resource utilisation play a decisive role in overall lifecycle performance \cite{du2023advances}. As infrastructure systems evolve toward greener and more sustainable development, there is an increasing demand for intelligent maintenance solutions that not only improve operational efficiency but also reduce energy usage and environmental impact \cite{jiao2023sustainable}. Therefore, enabling automated, energy-efficient operation for infrastructure systems has emerged as a significant and pressing engineering problem.

Real-world civil infrastructure encompasses a wide variety of operation and maintenance scenarios, among which several representative tasks require direct physical interaction with articulated components. Typical examples are opening access control doors to enter equipment rooms, pulling out control cabinets or service drawers for inspection, and turning pipeline valves to regulate flow in water supply, drainage, or fire protection systems \cite{liu2023literature}. These actions are fundamental steps in routine maintenance workflows and are essential for keeping infrastructure systems operating reliably. Together, they form a complete operational chain that spans equipment access, system inspection, and operating-parameter adjustment. Automating this chain can greatly improve maintenance efficiency and safety, and is a critical step toward truly intelligent and autonomous infrastructure management \cite{pregnolato2022towards}.

Although robotic technologies have been increasingly explored for infrastructure-related tasks, most existing methods primarily focus on grasping-centric manipulation, which mainly requires establishing stable contact between the gripper and an object \cite{billard2019trends}. In contrast, infrastructure maintenance often involves articulated object manipulation, where the robot must control the relative motion of connected components, such as doors, panels, or valves. Its difficulty is amplified by the sophisticated kinematic arrangements and nontrivial dynamic properties typical of articulated systems \cite{xie2023part}.

Current manipulation approaches can broadly be categorised into rule-based and learning-based methods. Rule-based manipulation relies heavily on predefined motion sequences and manually designed control strategies, which require extensive prior knowledge of object geometry and kinematic constraints \cite{chen2022review}. Such methods \cite{Zhao2020ModelBased} lack flexibility and adaptability, making them unsuitable for diverse and evolving infrastructure environments.

Learning-based robotic manipulation has therefore emerged as a promising alternative, particularly with the integration of visual perception enabled by low-cost and information-rich vision sensors \cite{ai2025review}. Several approaches \cite{Wang2019RLManipulation} attempt to reconstruct object geometry and estimate kinematic properties for manipulation planning, but these typically rely on interactive perception, making them unsuitable for untrained objects. Other methods \cite{wu2021vat} utilise visual affordance learning to guide manipulation strategies; however, these techniques still exhibit significant limitations, including strong reliance on expert knowledge for motion control, low sample efficiency in reinforcement learning (RL) algorithms, and poor generalisation across different articulated object categories with substantial shape variations.

{
Recent studies have advanced point-cloud–based perception for infrastructure systems, primarily targeting large-scale semantic segmentation and structural understanding. For example, Lin et al.~\cite{lin2025structure} proposed a structure-oriented loss function to improve semantic segmentation of bridge point clouds, while Jing et al.~\cite{jing2024lightweight} developed a lightweight Transformer-based network for large-scale masonry arch bridge point cloud segmentation. In contrast, our work addresses a different problem setting, focusing on state abstraction as the input to decision-making rather than on maximising perception accuracy.}

More critically, existing methods rarely account for energy consumption as a fundamental design factor, which limits their applicability in sustainable infrastructure operation. Furthermore, comprehensive theoretical guarantees regarding stability and convergence are often absent, raising concerns about their reliability in safety-critical, long-term deployment scenarios \cite{jain2025reinforcement}. Consequently, the suitability of current robotic solutions for infrastructure O\&M remains insufficient.

There is a clear engineering gap in robotic systems for intelligent infrastructure O\&M: existing methods are neither scalable nor energy-efficient enough to reliably handle diverse articulated components in real deployments. This calls for a unified, infrastructure-tailored framework that integrates robust perception, adaptive control, and energy-aware optimisation to enable sustainable, long-term operation.

{\color{black}
To summarise the current research on robotic manipulation for infrastructure operation and maintenance, existing approaches can be broadly characterised as follows:
\begin{itemize}
    \item \textbf{Grasp-centric or object-specific manipulation.}
    Most methods focus on grasping or rely on object-specific articulation and kinematic modelling, limiting scalability across diverse infrastructure components.
    
    \item \textbf{Limited generalisation across articulated objects.}
    Learning-based approaches often struggle to generalise to articulated components with varying geometries and motion constraints.
    
    \item \textbf{Energy treated as a secondary objective.}
    Actuation energy is commonly ignored or incorporated heuristically as a reward penalty rather than explicitly modelled as a constraint.
    
    \item \textbf{Insufficient reliability considerations.}
    Existing methods typically emphasise short-horizon task success, with limited attention to stability and constraint satisfaction in long-term, safety-critical deployment.
\end{itemize}
}

To address the above challenges, this work proposes an articulation-agnostic and energy-aware robotic manipulation framework specifically designed for intelligent infrastructure operation and maintenance tasks, as illustrated in Figure \ref{arch}. The framework integrates part-guided 3D perception with constrained RL for intelligent O\&M control. It enables a single policy to operate across diverse articulated infrastructure components, while explicitly regulating actuation energy consumption to support sustainable, long-term deployment.

{
\color{black}
The main contributions are summarized as follows:
\begin{itemize}
 \item \textbf{Articulation-agnostic manipulation framework for scalable infrastructure maintenance.}
{We develop a part-guided perception and control pipeline that formulates articulated-object manipulation through an articulation-agnostic state abstraction explicitly designed for constrained RL}, enabling a single control policy to handle diverse articulated components in intelligent infrastructure systems, such as access panels, control cabinets, and valve mechanisms. This reduces reliance on object-specific modelling and calibration, and improves adaptability and scalability for large-scale maintenance scenarios.

\item \textbf{Energy-aware constrained reinforcement learning for sustainable infrastructure operation.}
{We propose an energy-aware manipulation strategy by reformulating the control problem as a constrained Markov decision process, in which actuation energy is explicitly modeled as a constraint rather than a heuristic reward penalty.} This directly addresses the common neglect of action cost in task-driven manipulation, enabling more sustainable and cost-efficient operation of robotic systems in infrastructure environments.

\item \textbf{End-to-end multi-objective optimisation with theoretical reliability guarantees.}
{The proposed framework supports end-to-end learning by explicitly coupling perception-driven state abstraction with energy-aware constrained policy optimization within a unified primal--dual formulation, adaptively optimising task performance and energy efficiency.} The optimisation process is accompanied by a theoretical analysis of convergence and stability, which provides insight into the reliability of the proposed approach and supports its applicability to long-term robotic manipulation in infrastructure operation and maintenance scenarios.
\end{itemize}}

The remainder of this paper is organised as follows. 
Section~\ref{sec2} presents a comprehensive literature review of existing studies on robotic manipulation, articulated object interaction, and intelligent infrastructure operation and maintenance. 
Section~\ref{sec3} describes the proposed methodology, including the articulation-agnostic and energy-aware manipulation framework and its theoretical formulation. 
Section~\ref{sec4} details the system implementation, covering the environment setup, control architecture, and training procedures. 
Section~\ref{sec5} provides the experimental evaluation and discusses the performance of the proposed approach under representative infrastructure operation scenarios. 
Finally, Section~\ref{sec6} concludes key findings and future directions.

\section{Related Work}\label{sec2}
\subsection{Robotic Manipulation in Infrastructure Operation and Maintenance}

In recent years, robotic technologies have been increasingly explored for O\&M tasks in civil and infrastructure systems, motivated by the need to improve efficiency and reduce human exposure to hazardous environments. Zhang et al.~\cite{Zhang2019InfrastructureRobot} deployed mobile robots for automated inspection in subway tunnels, demonstrating improved efficiency in structural condition assessment and defect detection. Similarly, Li and Wang~\cite{Li2020BuildingRobot} proposed a robotic platform for building facility inspection, focusing on equipment monitoring, environmental surveillance, and routine condition assessment.

Chen et al.~\cite{Chen2021PipelineRobot} investigated pipeline maintenance robots for municipal water systems and highlighted their feasibility in confined, complex infrastructure environments. Although these systems have shown promising results in inspection and monitoring tasks, most existing robotic applications remain largely limited to passive sensing or simple interaction. Wang et al.~\cite{Wang2022MaintenanceLimitations} reported that current robotic systems still lack the robustness and adaptability required for complex physical manipulation operations, such as opening access doors, manipulating control cabinets, and regulating pipeline valves, thereby limiting their practical deployment in fully automated infrastructure O\&M workflows.

Overall, existing studies indicate that while robotic technologies have advanced infrastructure inspection and monitoring, their capacity to perform reliable, physically intensive manipulation tasks in real-world O\&M scenarios remains insufficient, highlighting the need for more adaptive and function-oriented manipulation frameworks.

\subsection{Articulated Object Manipulation Methods}

Existing approaches of articulated object manipulation can generally be categorised into rule-based and learning-based methods~\cite{Zhang2020ArticulatedSurvey}.

Rule-based approaches rely on explicit reconstruction of object geometry and estimation of kinematic parameters, such as joint axes, rotation centres, and motion constraints. Xu et al.~\cite{Xu2018KinematicsEstimation} proposed an interactive framework for estimating joint parameters in door manipulation tasks, while Zhao et al.~\cite{Zhao2020ModelBased} introduced geometry-driven control strategies based on detailed articulated object models. These methods enable accurate motion planning and precise control under well-defined conditions.

However, their effectiveness strongly depends on extensive prior knowledge, accurate object modelling, and repeated calibration processes. Such requirements are rarely satisfied in real infrastructure environments, where components often exhibit variability, wear, and unknown articulation properties. Consequently, model-based techniques lack scalability and are generally unsuitable for novel or unseen infrastructure components encountered in large-scale O\&M scenarios.

Learning-based manipulation approaches, particularly those based on deep RL and visual perception, have gained increasing attention for articulated object manipulation. Wang et al.~\cite{Wang2019RLManipulation} developed a visual RL framework for door opening, while Liu et al.~\cite{Liu2021Affordance} explored affordance-based learning strategies for manipulation in unstructured environments. These methods demonstrate the potential of learning policies directly from sensory observations, reducing the reliance on explicit object modelling.

Despite their progress, learning-based methods often exhibit low sample efficiency, strong dependence on expert-designed rewards, and limited generalisation capability across different articulated object categories. Li et al.~\cite{Li2022Generalisation} further highlighted that substantial shape variation and structural diversity significantly degrade performance when policies are transferred to unseen objects, thus restricting their applicability in real-world infrastructure settings.

\subsection{Functional Perception and Part-guided Manipulation}

To improve manipulation accuracy and task reliability, recent studies have explored functional perception approaches that focus on identifying and segmenting task-relevant object parts, such as handles, knobs, and control regions. Zhang et al.~\cite{Zhang2021PartSegmentation} proposed a deep part-segmentation framework to guide robotic interaction, while Wang et al.~\cite{Wang2020AffordanceDetection} introduced vision-based affordance detection techniques for identifying feasible manipulation regions.

These approaches have shown effectiveness in strengthening perception-guided manipulation performance, particularly in controlled or household environments. However, Chen et al.~\cite{Chen2022PartLimitations} noted that most existing part-guided methods remain highly object-specific and lack articulation-agnostic generalisation. Their deployment in complex infrastructure environments is therefore constrained, especially when faced with heterogeneous components exhibiting diversified articulation patterns and structural variations.

\subsection{Energy-aware and Constrained Reinforcement Learning}

Traditional RL frameworks for robotic manipulation predominantly focus on maximising task success or cumulative rewards, often neglecting physical costs such as energy consumption and actuator effort. Zhao et al.~\cite{Zhao2019EnergyBlindRL} demonstrated that reward-driven policies frequently result in energy-inefficient behaviour, which is undesirable for long-term autonomous operation in infrastructure systems.

To address this issue, Li et al.~\cite{Li2021ConstrainedRL} and Wang et al.~\cite{Wang2022EnergyAware} introduced Lagrangian-based constrained RL frameworks to incorporate resource and safety constraints into policy optimisation. These approaches enable the regulation of control effort and energy expenditure while maintaining task performance. Nevertheless, most existing studies remain confined to laboratory-scale experiments and generic manipulation scenarios, with limited validation in infrastructure O\&M environments that demand sustained, reliable, and energy-efficient operation.

To the best of our knowledge, we are the first to integrate functional part perception, articulation-agnostic manipulation, and energy-aware constrained optimisation into a unified end-to-end multi-objective RL framework specifically designed for intelligent infrastructure O\&M, enabling sustainable robotic manipulation under operational energy constraints.

\begin{figure*}
\centerline{\includegraphics[width=\linewidth,height=25pc]{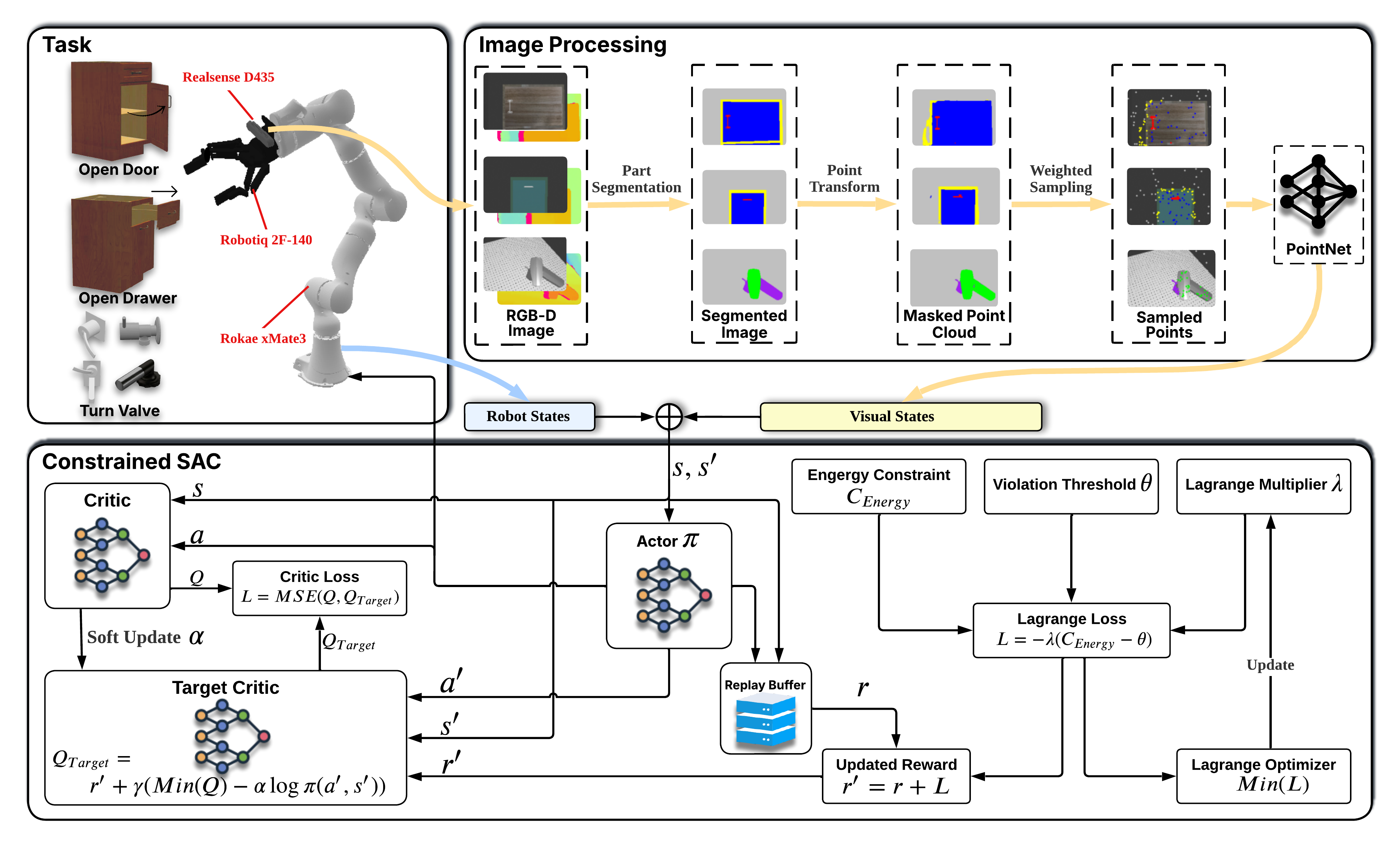}}
\caption{{Overview of the proposed energy-aware, articulation-agnostic, end-to-end RL manipulation framework. We integrate RGB-D part segmentation, masked point-cloud sampling, and PointNet-based visual encoding with robot proprioceptive states. A constrained SAC controller enforces an explicit energy constraint via a Lagrangian mechanism, allowing the policy to achieve effective articulated-object manipulation while regulating actuation energy consumption.
\label{archite}}}
\end{figure*}

Table~\ref{tab:related_work_summary} provides a concise comparison of representative studies. Prior approaches address only a subset of these aspects, whereas the proposed framework jointly integrates articulation-agnostic manipulation, part-guided perception, and energy-aware constrained optimisation within a unified end-to-end formulation.

\begin{table*}[t]
\centering
\color{black} 
\caption{{
Comparison of existing robotic manipulation approaches for infrastructure operation and maintenance.
Here, $\checkmark$ denotes explicit support, $\times$ denotes not supported, and $\sim$ denotes partial or implicit support.
}}
\label{tab:related_work_summary}
\small
\begin{tabular}{@{}p{0.28\linewidth} p{0.32\linewidth} c c c@{}}
\rowcolor{headergray}
\textbf{Category} &
\textbf{Representative Works} &
\textbf{Articulation} &
\textbf{Part-level} &
\textbf{Energy-aware} \\

Infrastructure inspection robots
&
\cite{Zhang2019InfrastructureRobot, Li2020BuildingRobot, Chen2021PipelineRobot}
&
$\times$
&
$\times$
&
$\times$
\\

\rowcolor{rowgray}
Rule-based articulated manipulation
&
\cite{Xu2018KinematicsEstimation, Zhao2020ModelBased}
&
$\times$
&
$\times$
&
$\times$
\\

Learning-based articulated manipulation
&
\cite{Wang2019RLManipulation, Liu2021Affordance, Li2022Generalisation}
&
$\sim$
&
$\sim$
&
$\times$
\\

\rowcolor{rowgray}
Functional / part-guided perception
&
\cite{Zhang2021PartSegmentation, Wang2020AffordanceDetection, Chen2022PartLimitations}
&
$\times$
&
$\checkmark$
&
$\times$
\\

Energy-aware constrained RL
&
\cite{Li2021ConstrainedRL, Wang2022EnergyAware}
&
$\sim$
&
$\times$
&
$\checkmark$
\\

\rowcolor{rowgray}
\textbf{Ours}
&
\textbf{This work}
&
$\checkmark$
&
$\checkmark$
&
$\checkmark$
\\
\end{tabular}
\end{table*}

\section{Methodology}\label{sec3}

\subsection{Overall Framework}
\label{overall}
The overall framework is shown in Figure~\ref{archite}. First, an RGB-D \cite{zhou2021rgb} perception stream performs part segmentation to identify functional regions such as handles, panels, and valves; the segmented regions are lifted into a 3D part-specific point cloud, from which weighted point sampling and a PointNet-based encoder \cite{qi2017pointnet} extract a compact geometric feature that captures local structure and functional affordances in an articulation-agnostic manner. Second, this geometric feature is concatenated with proprioceptive robot states to form a fused observation, which is consumed by an actor--critic control policy: the actor maps the fused representation to continuous joint-level actions, while paired critics estimate the expected task return and the associated energy cost, enabling a single policy to generalise across doors, drawers, and valves with diverse articulation structures. Third, the policy is optimised using a constrained Soft Actor-Critic (SAC) \cite{haarnoja2018soft} scheme, where an energy cost signal is evaluated during interaction and regulated via a Lagrange multiplier \cite{chen2025safety}, allowing the controller to adaptively balance task performance and energy efficiency for long-term infrastructure O\&M.

\subsection{Part-Guided Articulation-Agnostic Perception}
\label{sec:partguided}

We adopt a part-guided 3D perception pipeline to obtain an articulation-agnostic representation of articulated infrastructure components from RGB-D observations. The core idea is to segment objects into functionally meaningful rigid parts (e.g., handles, facades, bases), and to represent each part by a stable 3D point set with learned sampling weights, which can be directly consumed by the subsequent RL policy. Following the part representation paradigm in GAPartNet~\cite{geng2023gapartnet}, we define a part as a rigid segment that shares similar affordances and supports consistent interaction behaviours across different articulated objects. For example, cabinet doors and drawers are decomposed into three parts (handle, facade, and base), while valves consist of a handle and a fixed base. Such an object-irrelevant, affordance-oriented part definition facilitates generalisation across diverse articulated mechanisms.

Given an RGB-D observation at time step $t$,
\begin{equation}
    I_t = \bigl\{ I^{\text{rgb}}_t,\; I^{\text{depth}}_t \bigr\}
\end{equation}
a segmentation network $f_{\text{seg}}$ predicts a pixel-wise part label map
\begin{equation}
    M_t = f_{\text{seg}}\!\bigl(I^{\text{rgb}}_t\bigr)
\end{equation}
where $M_t(u,v) \in \{1,\dots,K\}$ denotes the part index at pixel $(u,v)$. Parts are defined to be agnostic to object identity or category, so that the same label (e.g., ``handle'') is shared by handles on doors, drawers, and valves. To train $f_{\text{seg}}$, we generate a large-scale synthetic dataset of articulated objects in simulation. Objects from multiple categories are placed on a table, and a hand-centric RGB-D sensor mounted on a floating gripper observes them from diverse viewpoints. Ground-truth part masks are obtained directly from the simulator, without manual annotation. Scene-level randomisation is applied to enrich the dataset, including variations in object pose, articulation angle, opening distance, camera pose, and gripper--object distance. 

Using the depth image and calibrated camera intrinsics $K$, pixels belonging to a given part label $\ell$ are lifted to the world frame. Let
\begin{equation}
    \Omega_t^{\ell} = \bigl\{ (u,v) \,\big|\, M_t(u,v) = \ell \bigr\}
\end{equation}
denote the pixel set of part $\ell$. Each pixel $(u,v)\in\Omega_t^{\ell}$ is transformed into a 3D point \cite{li2017improved}
\begin{equation}
    \mathbf{p}_{uv} = D_t(u,v)\, K^{-1}
    \begin{bmatrix}
        u \\ v \\ 1
    \end{bmatrix}
\end{equation}
and the resulting part-specific point set is
\begin{equation}
    \mathcal{P}_t^{\ell} = \bigl\{ \mathbf{p}_{uv} \,\big|\, (u,v)\in\Omega_t^{\ell} \bigr\}
\end{equation}
Since the raw point clouds can be large and imbalanced across parts, downsampling is required to make subsequent RL training efficient while preserving the geometric structure of small but functionally critical parts.

To obtain a compact yet reliable point set for each part, we adopt a frame-consistent uncertainty-aware sampling strategy. The goal is to sample a fixed number $N_s$ of points per part, such that they (i) are more likely to come from reliable segmentation regions and (ii) remain consistent across consecutive frames during manipulation. First, we estimate per-pixel segmentation uncertainty via test-time augmentation \cite{shanmugam2020and} and Monte Carlo dropout \cite{camarasa2020quantitative}. For a given RGB-D input $I_t$, we apply $K$ stochastic augmentations and forward passes through $f_{\text{seg}}$, obtaining softmax probability maps \cite{pearce2021understanding}
\begin{equation}
    \bigl\{ P^{(k)}_t \in \mathbb{R}^{K \times H \times W} \bigr\}_{k=1}^{K}
\end{equation}
The mean class probability at each pixel is
\begin{equation}
    \bar{P}_{t,c}(u,v) = \frac{1}{K} \sum_{k=1}^{K} P^{(k)}_{t,c}(u,v)
\end{equation}
and the predictive entropy \cite{hernandez2016predictive}
\begin{equation}
    U_t(u,v) = - \sum_{c=1}^{K} \bar{P}_{t,c}(u,v)\, \log \bar{P}_{t,c}(u,v)
\end{equation}
serves as an uncertainty map, which is normalised to $[0,1]$. For each part $\ell$, the uncertainty scores of its pixels form a vector
\begin{equation}
    \mathbf{U}_t^{\ell} = \bigl\{ U_t(u,v) \,\big|\, (u,v)\in\Omega_t^{\ell} \bigr\}
\end{equation}
from which uncertainty-based sampling weights
\begin{equation}
    \mathbf{w}^{\ell}_{\text{ua}} = \text{softmax}\!\bigl( \mathbf{U}_t^{\ell} \bigr)
\end{equation}
are derived. These weights bias sampling toward pixels with higher predictive uncertainty, which empirically helps mitigate segmentation errors in the employment process. {In articulated infrastructure components, regions with higher predictive uncertainty frequently correspond to part boundaries, joints, or contact interfaces, which are critical for inferring motion constraints and manipulation direction. Accordingly, uncertainty-aware sampling is used to improve geometric coverage of such informative regions.}

To improve temporal alignment of the sampled points, we incorporate frame-consistent weights. For each part $\ell$, we maintain a queue $Q^{\ell}$ of $N_s$ points sampled in previous frames. At time $t$, the distance between each current point $\mathbf{p}\in\mathcal{P}_t^{\ell}$ and the closest point in $Q^{\ell}$ is computed, yielding a distance vector
\begin{equation}
    \mathbf{d}_t^{\ell} = 
    \Bigl\{ 
        \min_{\mathbf{q} \in Q^{\ell}} \bigl\| \mathbf{p} - \mathbf{q} \bigr\|_2
        \;\Big|\; \mathbf{p} \in \mathcal{P}_t^{\ell}
    \Bigr\}
\end{equation}
Frame-consistent weights are then defined as
\begin{equation}
    \mathbf{w}^{\ell}_{\text{fc}} = 2^{-k_{\text{fc}} \mathbf{d}_t^{\ell}}
\end{equation}
where $k_{\text{fc}} > 0$ is a decay coefficient. {
Here, $\mathbf{d}_t^{\ell} = \{ d_{t,i}^{\ell} \}_{i=1}^{|\mathcal{P}_t^{\ell}|}$ denotes a vector of per-point distances, and the exponential operation is applied in an element-wise manner.} Points closer to previously sampled ones receive larger weights, promoting temporal stability and filtering out noisy or unstable points. The combined sampling weights for part $\ell$ are given by element-wise multiplication
\begin{equation}
    \mathbf{w}^{\ell} = \mathbf{w}^{\ell}_{\text{ua}} \circ \mathbf{w}^{\ell}_{\text{fc}}
\end{equation}
and $N_s$ points are sampled from $\mathcal{P}_t^{\ell}$ according to $\mathbf{w}^{\ell}$. \textcolor{black}{Here, the operator ``$\circ$'' denotes the Hadamard (element-wise) product, which serves as a simple and effective ``AND''-style gating mechanism: points are emphasised only when they are both spatially reliable and temporally stable. This fixed fusion avoids introducing additional learnable parameters, thereby mitigating noise amplification and improving stability in RL.
} The sampled points from all parts are concatenated to form the final point set
\begin{equation}
    \mathcal{P}^*_t = \bigcup_{\ell=1}^{K} \hat{\mathcal{P}}_t^{\ell}
\end{equation}
which is then fed into a PointNet-based encoder to obtain a compact geometric feature for the downstream control policy. {We emphasise that the role of the geometric encoder in this work is not to maximise 3D perception accuracy, but to provide a compact and control-oriented state abstraction for decision-making in reinforcement learning. Accordingly, we adopt PointNet as a lightweight geometric encoder due to its linear computational complexity, deterministic inference latency, and stable optimisation behaviour, which are particularly important for long-horizon closed-loop manipulation in infrastructure operation and maintenance scenarios, when compared with more computationally demanding alternatives such as PointNet++ \cite{qi2017pointnet} and transformer-based \cite{jing2024lightweight} models.}

This part-guided and sampling-aware 3D perception pipeline yields an object-irrelevant, function-centric representation that is robust to segmentation noise, temporally consistent during manipulation, and scalable across different articulated infrastructure components.

\subsection{Energy-Aware Constrained Reinforcement Learning Formulation}

\label{sec:csac_formulation}

Robotic manipulation for infrastructure operation and maintenance inherently involves a trade-off between task effectiveness and energy expenditure \cite{zhao2011joint}.  
High-torque or rapidly changing control actions may complete a task efficiently but lead to undesirable energy usage, which is incompatible with the long-term, cost-sensitive nature of real-world infrastructure O{\&}M.  
To encode this engineering requirement directly into the decision-making process, we formulate the manipulation problem as a \emph{Constrained Markov Decision Process} (CMDP) \cite{altman2021constrained}.

Formally, we define the CMDP as
\begin{equation}
    \mathcal{M} = (\mathcal{S}, \mathcal{A}, P, r, c, \gamma)
\end{equation}
where $\mathcal{S}$ is the state space (including the articulation-agnostic geometric feature $\mathbf{g}_t$ and proprioception),  
$\mathcal{A}$ is the continuous action space,  
$P$ is the transition kernel,  
$r$ denotes task rewards that measure manipulation performance,  
$c$ denotes the instantaneous energy consumption,  
and $\gamma \in (0,1)$ is the discount factor \cite{tessler2018reward}.

The control objective \cite{achiam2017constrained} is to maximise the expected task return while ensuring that the long-term expected energy consumption does not exceed a prescribed engineering budget $\theta$:
\begin{equation}
    \max_{\pi} \; J_r(\pi)
    \quad \text{s.t.} \quad
    J_c(\pi) \le \theta 
    \label{eq:cmdp}
\end{equation}
where
\begin{align}
    J_r(\pi) &= 
    \mathbb{E}_{\pi}\!\left[ 
        \sum_{t=0}^{\infty} \gamma^t r(s_t,a_t)
    \right] \\
    J_c(\pi) &= 
    \mathbb{E}_{\pi}\!\left[ 
        \sum_{t=0}^{\infty} \gamma^t c(s_t,a_t)
    \right]
\end{align}
The constraint $J_c(\pi)\le\theta$ ensures that the learned controller respects a predefined energy budget, enabling sustain\-able operation across long-horizon O{\&}M tasks.

To solve~\eqref{eq:cmdp}, we adopt a Lagrangian relaxation, transforming the constrained optimisation into an unconstrained saddle-point problem \cite{huang2023unified}:
\begin{equation}
    \mathcal{L}(\pi,\lambda)
    = 
    J_r(\pi)
    - \lambda \bigl( J_c(\pi) - \theta \bigr)
    \label{eq:lagrangian}
\end{equation}
where $\lambda\ge 0$ is the Lagrange multiplier that adaptively adjusts the strength of the energy penalty.  
Unlike heuristic reward shaping, this formulation provides a principled mechanism for enforcing long-term energy constraints and is particularly suitable for engineering control systems where operating costs must be regulated explicitly.

We instantiate this Lagrangian formulation using a constrained SAC, employing separate critics for reward and energy cost,
\begin{equation}
Q_r(s,a) \quad \text{and} \quad Q_c(s,a)
\end{equation}

For a transition $(s_t, a_t, r_t, c_t, s_{t+1})$ sampled from the replay buffer $D$, 
the reward critic uses the entropy-regularised Bellman backup \cite{sutton1999reinforcement}:
\begin{equation}
    y_r =
    r_t
    +
    \gamma \,
    \mathbb{E}_{a' \sim \pi(\cdot \mid s_{t+1})}
    \bigl[
        Q_r^{\mathrm{tgt}}(s_{t+1},a')
        -
        \alpha \log \pi(a' \mid s_{t+1})
    \bigr]
    \label{eq:reward_backup}
\end{equation}
while the cost critic uses the standard discounted backup:
\begin{equation}
    y_c =
    c_t
    +
    \gamma \,
    \mathbb{E}_{a' \sim \pi(\cdot \mid s_{t+1})}
    \bigl[
        Q_c^{\mathrm{tgt}}(s_{t+1},a')
    \bigr]
    \label{eq:cost_backup}
\end{equation}
The reward and cost critics are then updated by minimising their respective Bellman residuals:
\begin{align}
    \mathcal{J}_{Q_r}
    &=
    \mathbb{E}_{(s_t,a_t,r_t,c_t,s_{t+1})\sim D}
    \bigl[
        \bigl( Q_r(s_t,a_t) - y_r \bigr)^2
    \bigr]
    \\
    \mathcal{J}_{Q_c}
    &=
    \mathbb{E}_{(s_t,a_t,r_t,c_t,s_{t+1})\sim D}
    \bigl[
        \bigl( Q_c(s_t,a_t) - y_c \bigr)^2
    \bigr]
\end{align}

Given replay samples, the actor is updated by minimising
\begin{equation}
    \mathcal{J}_{\pi}
    =
    \mathbb{E}_{s_t \sim D,\; a_t \sim \pi(\cdot \mid s_t)}
    \Bigl[
        \alpha \log \pi(a_t \mid s_t)
        - Q_r(s_t,a_t)
        + \lambda\, Q_c(s_t,a_t)
    \Bigr]
    \label{eq:actor_loss}
\end{equation}
where $\alpha$ is the entropy temperature.  
The term $\lambda\, Q_c(s_t,a_t)$ explicitly penalises energy-inefficient actions, guiding the policy toward solutions that balance manipulation performance and sustainable energy usage.  
Crucially, the articulation-agnostic geometric feature $\mathbf{g}_t$ produced by the perception module is embedded directly into the state representation $s_t$, enabling the constrained policy to generalise across different articulated infrastructure components without requiring object-specific models or kinematic priors.

This formulation provides the theoretical foundation for the proposed energy-aware manipulation framework.  
The optimisation dynamics and stability properties of the primal--dual updates are discussed in Section~\ref{sec:optimization}.

\subsection{Optimisation Strategy and Theoretical Properties}
\label{sec:optimization}

The constrained SAC framework in Section~\ref{sec:csac_formulation} is optimised via a primal-dual procedure that alternates between policy improvement, value estimation, and adaptive constraint regulation. This section outlines the optimisation strategy and highlights theoretical properties that support stable and reliable behaviour in long-horizon infrastructure O{\&}M scenarios.

To enforce the long-term energy constraint $J_c(\pi)\le\theta$, the Lagrange multiplier $\lambda$ is updated through a dual ascent step:
\begin{equation}
    \lambda \leftarrow 
    \max\!\Bigl( 0,\,
        \lambda + \eta_\lambda \bigl( J_c(\pi) - \theta \bigr)
    \Bigr)
    \label{eq:lambda_update}
\end{equation}
where $\eta_\lambda>0$ is the dual learning rate.  
When the expected energy consumption exceeds the budget $\theta$, $\lambda$ increases, amplifying the penalty on high-energy actions in the actor objective~\eqref{eq:actor_loss};  
conversely, when $J_c(\pi)<\theta$, $\lambda$ decreases, allowing the policy to place more emphasis on task performance.  
This adaptive mechanism steers the optimisation towards a feasible solution of the constrained problem~\eqref{eq:cmdp} and implements a standard dual ascent scheme for CMDPs.

Both the reward critic $Q_r$ and the cost critic $Q_c$ employ target networks with Polyak averaging to stabilise temporal-difference estimation:
\begin{equation}
    \phi_{\text{target}}
    \leftarrow
    \tau \phi
    +
    (1-\tau)\phi_{\text{target}}
    \quad 0 < \tau \ll 1
    \label{eq:polyak}
\end{equation}
where $\phi$ denotes critic parameters.  
Soft target updates reduce the variance and drift of bootstrapped value targets, which in turn prevents oscillatory behaviour in policy updates and mitigates the risk of policy collapse.  
In the context of infrastructure O{\&}M, such stabilisation is crucial, as large fluctuations in the learned policy could translate into unsafe or inefficient behaviours during real deployment.

We provide a Lyapunov-style \cite{khalil2009lyapunov} argument to characterise the convergence behaviour of the proposed primal--dual updates.

\begin{definition}[Energy-feasible policy]
A policy $\pi$ is said to be \emph{energy-feasible} if its long-term expected energy consumption satisfies
\begin{equation}
    J_c(\pi) \;\le\; \theta
\end{equation}
\end{definition}

Let $\pi^\star$ denote an optimal energy-feasible policy for the CMDP~\eqref{eq:cmdp}.  
Consider the following Lyapunov candidate over the iterates $(\pi_k,\lambda_k)$:
\begin{equation}
    \mathcal{V}_k 
    =
    \bigl( J_r(\pi^\star) - J_r(\pi_k) \bigr)
    +
    \frac{1}{2\rho}
    \Bigl( \bigl[J_c(\pi_k) - \theta\bigr]_+ \Bigr)^2
    \label{eq:lyapunov}
\end{equation}
where $\rho>0$ is a weighting parameter and $[x]_+=\max(x,0)$.

\begin{proposition}[Lyapunov decrease and asymptotic constraint satisfaction]
\label{prop:lyapunov}
Suppose that $J_r(\pi)$ and $J_c(\pi)$ are bounded and Lipschitz-continuous with respect to the policy parameters, and that the actor and dual step sizes $\eta_\pi$ and $\eta_\lambda$ are chosen sufficiently small.  
Then the stochastic primal--dual updates associated with~\eqref{eq:actor_loss} and~\eqref{eq:lambda_update} admit the Lyapunov function $\mathcal{V}_k$ in~\eqref{eq:lyapunov}, and there exists a constant $C>0$ such that
\begin{equation}
    \frac{1}{T} \sum_{k=0}^{T-1}
    \bigl( J_c(\pi_k) - \theta \bigr)_+
    \;\le\;
    \frac{C}{\sqrt{T}}
    \quad \text{for all } T \ge 1
    \label{eq:constraint_bound}
\end{equation}
In particular, the average constraint violation vanishes as $T\to\infty$, and any limit point of $\{\pi_k\}$ is energy-feasible.
\end{proposition}

\noindent\textit{Proof (sketch).}
Under the stated regularity assumptions, the actor and dual updates implement a stochastic gradient descent--ascent step on the Lagrangian $\mathcal{L}(\pi,\lambda)$ in~\eqref{eq:lagrangian}.  
Expanding $\mathcal{V}_{k+1}-\mathcal{V}_k$ along the update directions and taking expectation yields
\begin{equation}
\begin{aligned}
    \mathbb{E}\!\left[ \mathcal{V}_{k+1} - \mathcal{V}_k \right]
    \le\;&
    - \eta_\pi 
        \bigl\| \nabla_\pi \mathcal{L}(\pi_k,\lambda_k) \bigr\|^2
    - \eta_\lambda 
        \bigl( J_c(\pi_k) - \theta \bigr)_+^2
    \\
    &
    + \mathcal{O}(\eta_\pi^2,\, \eta_\lambda^2)
\end{aligned}
\label{eq:lyapunov_decrease}
\end{equation}

showing that $\mathcal{V}_k$ is a supermartingale up to higher-order terms.  
Summing~\eqref{eq:lyapunov_decrease} over $k$ and applying standard stochastic approximation arguments leads to the bound~\eqref{eq:constraint_bound} and to the vanishing of both the primal gradient norm and the constraint violation in the limit.

Taken together, the dual ascent update~\eqref{eq:lambda_update}, 
the stabilised critic dynamics~\eqref{eq:polyak}, 
and the Lyapunov-style decrease property~\eqref{eq:lyapunov_decrease} 
provide a theoretical underpinning for the proposed energy-aware control scheme.  
They indicate that, under mild regularity conditions, the combined primal--dual updates 
converge towards a stationary solution of the CMDP with vanishing constraint violation 
and bounded energy consumption, thereby offering strong reliability guarantees for long-term deployment 
in intelligent infrastructure operation and maintenance.

\section{Implementation}\label{sec4}

\subsection{Environment Setup}
\label{sec:env_setup}

The proposed framework is implemented and evaluated on a 7-DOF manipulator platform and a physics-based simulation environment. The physical setup consists of an optical table, a ROKAE xMate3 Pro robot arm \cite{rokae2020xmate3pro} equipped with a Robotiq 2F-140 two-finger gripper \cite{robotiq2f140} mounted at the end-effector, and an Intel RealSense D435 RGB-D camera \cite{intelrealsenseD435}. The camera is rigidly attached in a hand-centric configuration, providing depth and colour observations of the manipulated object and the end-effector. {During experiments, both the robotic manipulator and the articulated objects are placed on the optical table, and the objects are randomly positioned for each trial. The training and testing environments share the same setup and configurations to ensure consistent and fair evaluation.}

For training and data collection, we use a physics simulator to emulate articulated object manipulation. {As representative objects for O\&M tasks, we select doors, drawers, and valves as the articulated objects from the publicly available PartNet-Mobility dataset~\cite{Xiang_2020_SAPIEN}.} These assets are used to generate synthetic RGB-D observations and part annotations for training the perception module. In the RL stage, we instantiate 4 object instances from each category as training environments. 

The visual observation at each time step consists of a single hand-centric RGB-D frame with a resolution of $144 \times 256$ pixels. To enable fast segmentation inference compatible with RL training, we adopt a streamlined encoder–decoder architecture inspired by U-Net~\cite{ronneberger2015u}, using MobileNetV2~\cite{sandler2018mobilenetv2} as the encoder backbone to balance feature expressiveness and computational efficiency.

All experiments are conducted on a workstation equipped with an i9-12950HX CPU  and an NVIDIA GeForce RTX A4500 GPU running Ubuntu 24.04. The physics simulation and all models are implemented in Python~3.8 with PyTorch~2.4.1 \cite{paszke2019pytorch} and gym~0.21.0.

\subsection{RL MDP Details}
\label{sec:mdp_details}

The manipulation task is instantiated as the CMDP described in Section~\ref{sec:csac_formulation}. 
Here, we provide the explicit CMDP components used during training.

\textbf{State Space.}
The policy receives the fused observation
\begin{equation}
s_t = \bigl[\, \mathbf{g}_t,\; q_t,\; x^{\text{ee}}_t \,\bigr]
\end{equation}
where $\mathbf{g}_t$ is the geometric feature produced by the part-guided perception module,  
$q_t$ is the position of 7 joints,  
and $x^{\text{ee}}_t$ denotes the end-effector pose. {Here, proprioception refers to the internal kinematic state of the robot, including joint-level configuration and the Cartesian pose of the end-effector.}

\textbf{Action Space.}
The action is defined as the robot target joint positions together with the gripper finger position:
\begin{equation}
    a_t = \bigl[\, q^{\text{tar}}_t,\; f_t \,\bigr]
\end{equation}
where $q^{\text{tar}}_t$ denotes the target joint positions of the 7-DOF arm, and $f_t$ is the target position of the gripper finger.  
{All action dimensions are normalised to a fixed range before being sent to the low-level Proportional–Integral–Derivative (PID) controller \cite{johnson2005pid}.}

\textbf{Task Reward.}
To guide the policy towards physically meaningful articulated manipulation, the task reward combines several shaping terms that reflect different aspects of a successful operation.  
At each time step $t$, the scalar reward is defined as
\begin{equation}
    r_t
    =
    \alpha_{\text{reach}}\, r_t^{\text{reach}}
    + \alpha_{\text{align}}\, r_t^{\text{align}}
    + \alpha_{\text{goal}}\, r_t^{\text{goal}}
    + \alpha_{\text{vis}}\, r_t^{\text{vis}}
    + \alpha_{\text{grasp}}\, r_t^{\text{grasp}}
\end{equation}
where $r_t^{\text{reach}}$ encourages the end-effector to move closer to the target movable part, 
$r_t^{\text{align}}$ promotes motion along the desired opening or turning direction, 
$r_t^{\text{goal}}$ rewards progress towards the final task configuration (e.g., fully opened door or valve angle), 
$r_t^{\text{vis}}$ favours keeping the manipulated part within the camera field of view, 
and $r_t^{\text{grasp}}$ reflects the quality of the physical contact when grasping is required.  
The coefficients $\alpha_{\text{reach}}$, $\alpha_{\text{align}}$, $\alpha_{\text{goal}}$, $\alpha_{\text{vis}}$, and $\alpha_{\text{grasp}}$ weight the relative importance of these components.  
For the \emph{TurnValve} scenario, the explicit grasping component is disabled by setting $\alpha_{\text{grasp}} = 0$, since the valve can be operated without a dedicated grasping phase.

\textbf{Energy Cost.}
To quantify the actuation effort, we define the instantaneous energy-related cost at time $t$ as
\begin{equation}
    c_t
    =
    \beta_{\text{eng}} \,\bigl\| \dot{q}_t \bigr\|_2^2
\end{equation}
where$\beta_{\text{eng}}>0$ is a scaling coefficient.  
{This term serves as a smooth proxy for actuation effort and is adopted for its numerical stability and suitability for long-horizon constrained policy optimisation.} It is used as the cost function $c(s_t,a_t)$ in the CMDP formulation, whose discounted expectation $J_c(\pi)$ is constrained to satisfy the energy budget $J_c(\pi)\le\theta$ as defined in Section~\ref{sec:csac_formulation}. \textcolor{black}{Note that this cost is a scaled proxy, and the absolute cost values are not directly comparable across different robotic actuation platforms without appropriate calibration. Accordingly, the main claims in this work are based on comparative evaluations across methods under the same experimental setup.}


\subsection{Model Architecture Details}
\label{sec:architecture_details}

The full control pipeline follows the constrained SAC design in Section~\ref{sec:csac_formulation}, with lightweight MLP-based networks for real-time control. The main architectural components are summarised in Table~\ref{tab:arch}. {We adopt the standard SAC automatic entropy tuning mechanism, which adjusts the entropy temperature online to maintain a target entropy and balance exploration and exploitation without manual tuning.}

\renewcommand{\arraystretch}{1.25} 

\begin{table*}[t]
\centering
\caption{Network architecture configuration.}
\label{tab:arch}

\begin{tabular}{@{}p{0.28\linewidth} p{0.68\linewidth}@{}}
\rowcolor{headergray}
\textbf{Component} & \textbf{Configuration} \\

Geometric encoder &
PointNet: shared MLP (64, 128, 256) $\rightarrow$ max pooling $\rightarrow$ FC 256 $\rightarrow$ 128 (128-d feature $\mathbf{g}_t$) \\
\rowcolor{rowgray}
Actor network &
2-layer MLP (256, 256); outputs Gaussian mean and log-variance with $\tanh$-squashed actions \\

Reward critic $Q_r(s,a)$ &
Twin 2-layer MLPs (256, 256); separate target networks updated via Polyak averaging \\ 
\rowcolor{rowgray}
Cost critic $Q_c(s,a)$ &
Twin 2-layer MLPs (256, 256); same structure and update scheme as $Q_r$ \\

Entropy coefficient &
Automatic entropy tuning (SAC-style temperature adaptation) \\
\rowcolor{rowgray}
Lagrange multiplier $\lambda$ &
Initial value $\lambda_0 = 1.0$; updated online via dual ascent as in Eq.~\eqref{eq:lambda_update} \\

\end{tabular}

\begin{tablenotes}
\item Notes: All MLPs use ReLU nonlinearities and layer normalisation where appropriate. 
Actor and critic architectures are kept lightweight to enable real-time deployment on infrastructure O\&M platforms.
\end{tablenotes}

\end{table*}


\subsection{Training Schema}
\label{sec:training_schema}
 Table \ref{tab:hyper} lists Constrained SAC hyperparameters shared
across methods.
The complete training pipeline follows the Constrained SAC optimisation strategy defined in Section~\ref{sec:optimization}. {We set $\alpha_{\text{reach}}=2.0$, $\alpha_{\text{align}}=20.0$, $\alpha_{\text{goal}}=10.0$, $\alpha_{\text{vis}}=1.0$, and $\alpha_{\text{grasp}}=1.0$ to balance the numerical scales of heterogeneous reward components, whose raw magnitudes differ substantially across distance-, angle-, and stage-based terms. \textcolor{black}{We set $N_s = 32$ points per functional part, which provides sufficient geometric coverage of each part for PointNet while maintaining computational efficiency.} We set $K = 10$ as it provides a stable estimate of predictive uncertainty with negligible computational overhead, and $k_{\text{fc}} = 40$ as it strongly favours temporally consistent points across consecutive frames, improving stability during manipulation.
Higher weights are assigned to task-critical components, such as alignment and goal completion, to facilitate effective credit assignment during learning, while auxiliary terms are intentionally assigned lower weights. The part segmentation network is trained offline using synthetic annotations and remains fixed during policy learning.
During RL training, gradients are not back-propagated through the perception module.}
{Algorithm~\ref{alg:methodology_training} summarises the end-to-end training procedure of the proposed energy-aware framework. In implementation, all components are executed within a single training loop. At each time step, the perception module provides a geometry-based state representation, the actor outputs control actions for the robotic manipulator, and the collected reward and energy cost are used to update the policy, critics, and Lagrange multiplier in a fully end-to-end manner.}

{Importantly, all parameters are kept fixed across tasks and environments, and no task-specific tuning is performed. The architectures and hyperparameters follow standard RL practices, ensuring stable training and a fair evaluation of the proposed formulation.} This training scheme jointly optimises task performance and long-term energy efficiency, enabling reliable and sustainable manipulation across diverse articulated infrastructure components.

\begin{algorithm}
\caption{End-to-End RL Training of the Proposed Energy-Aware Articulation-Agnostic Manipulation Framework}
\label{alg:methodology_training}
\begin{algorithmic}[1]
\State \textbf{Input:} segmentation network $f_{\text{seg}}$; geometric encoder $f_{\text{geo}}$;
actor $\pi_\theta$; reward critic $Q_r$; cost critic $Q_c$;
Lagrange multiplier $\lambda$; replay buffer $\mathcal{D}$.

\State \textbf{Initialize} target critics $Q_r^{\text{tgt}}$, $Q_c^{\text{tgt}}$.

\For{each training iteration}
    \State \textbf{(Perception)} Capture RGB-D observation $I_t$.
    \State Predict part segmentation mask: $M_t = f_{\text{seg}}(I^{\text{rgb}}_t)$.
    \State Lift masked depth pixels to 3D point sets $\{ \mathcal{P}_t^\ell \}$.
    \State Compute uncertainty and frame-consistency weights.
    \State Sample per-part weighted points $\widehat{\mathcal{P}}_t^\ell$ and fuse:
    \[
        \mathcal{P}^*_t = \bigcup_{\ell} \widehat{\mathcal{P}}_t^\ell .
    \]
    \State Extract geometric feature $\mathbf{g}_t = f_{\text{geo}}(\mathcal{P}^*_t)$.
    \State Form policy observation $s_t = \bigl(\mathbf{g}_t,\text{proprioception}\bigr)$.
    
    \State \textbf{(Action Selection)} Sample action $a_t \sim \pi_\theta(a|s_t)$.
    \State Execute $a_t$, and observe $(r_t, c_t, s_{t+1})$.
    \State Store transition $(s_t,a_t,r_t,c_t,s_{t+1})$ into $\mathcal{D}$.

    \State \textbf{(Critic Updates)} Sample minibatch from $\mathcal{D}$.
    \State Compute reward target $y_r$ using Eq.~\eqref{eq:reward_backup}.
    \State Compute cost target $y_c$ using Eq.~\eqref{eq:cost_backup}.
    \State Update $Q_r$  by minimising Bellman residual:
\[
\mathcal{J}_{Q_r}
=
\mathbb{E}\bigl[
    \bigl(Q_r(s_t,a_t) - y_r\bigr)^2
\bigr].
\]

\State Update $Q_c$ by minimising Bellman residual:
\[
\mathcal{J}_{Q_c}
=
\mathbb{E}\bigl[
    \bigl(Q_c(s_t,a_t) - y_c\bigr)^2
\bigr].
\]
    \State Update target networks:
    \[
    \phi_{\mathrm{tgt}} \leftarrow \tau \phi + (1-\tau)\phi_{\mathrm{tgt}}.
    \]

    \State \textbf{(Actor Update)} Update policy via
    \[
    \nabla_\theta \mathcal{J}_\pi
    =
    \nabla_\theta 
    \mathbb{E}_{s\sim D,\,a\sim\pi_\theta}
    \Bigl[
        \alpha \log\pi_\theta(a|s)
        - Q_r(s,a)
        + \lambda Q_c(s,a)
    \Bigr].
    \]

    \State \textbf{(Dual Update)} Update $\lambda$ using minibatch estimate $\hat{J}_c$:
    \[
    \lambda \leftarrow 
        \max\!\Bigl(0,\;
        \lambda + \eta_\lambda(\hat{J}_c - \theta)
    \Bigr).
    \]

\EndFor
\end{algorithmic}
\end{algorithm}

\begin{table}[t]
\centering
\caption{Constrained SAC training hyperparameters.}
\label{tab:hyper}
\begin{tabular}{@{}p{0.5\linewidth}p{0.45\linewidth}@{}}
\rowcolor{headergray}
\textbf{Hyperparameter} & \textbf{Value} \\

Policy Learning Rate & 4e-4 \\
\rowcolor{rowgray}
Value Learning Rate & 4e-4 \\
Lagrange Learning Rate & 1e-3 \\
\rowcolor{rowgray}
Initial Lagrange Multiplier $\lambda$ & 1.0 \\
Initial Violation Threshold $\theta$ & 0.06 \\
\rowcolor{rowgray}
Batch Size & 256 \\
Buffer Size & 100{,}000 \\
\rowcolor{rowgray}
Discount Factor $\gamma$ & 0.99 \\
Soft Update Coefficient $\tau$ & 0.005 \\
\rowcolor{rowgray}
Optimizer & Adam \\
\end{tabular}

\end{table}

\section{Experimental Evaluation}
\label{sec5}

\subsection{Evaluation Setup}

\textbf{Research Questions}
Our experimental evaluation is organized to address the following research questions (RQs):
\begin{itemize}
    \item \textbf{RQ1 (Energy efficiency).} Does the proposed energy-aware constrained RL framework reduce the total energy consumption required to complete articulated-object manipulation tasks?
    \item \textbf{RQ2 (Manipulation efficiency).} Does enforcing an energy constraint lead to shorter and smoother manipulation trajectories, as reflected by the total number of action steps to success?
    \item \textbf{RQ3 (Task reliability).} Does incorporating an explicit energy constraint affect the success rate of completing real-world–relevant articulated-object operations?
    \item \textbf{RQ4 (Constraint satisfaction).} During training, does the policy consistently satisfy the long-horizon energy constraint imposed by the violation threshold~$\theta$?
\end{itemize}

\textbf{Evaluation Metrics}
We evaluate each method using three task-oriented metrics:
\begin{itemize}
    \item \textbf{Total Energy Cost.} The accumulated energy cost over an entire episode, computed using the energy model described in Section~\ref{sec:mdp_details}. This metric quantifies the effectiveness of energy-aware policy shaping.
    \item \textbf{Steps to Completion.} The number of control steps required to successfully complete the task, reflecting manipulation efficiency and smoothness of the executed trajectory.
    \item \textbf{Success Rate.} The fraction of evaluation episodes in which the robot successfully completes the articulated manipulation (door opening, drawer pulling, or valve rotation).
\end{itemize}

\textbf{Task Scenarios}
All methods are evaluated on three representative infrastructure O\&M manipulation tasks, as illustrated in Figures \ref{fig:env_repre} and \ref{fig:dataset}:
\begin{itemize}
    \item \textbf{OpenDoor}: rotating a hinged door around its revolute joint.
    \item \textbf{OpenDrawer}: pulling out a sliding drawer along its prismatic joint.
    \item \textbf{TurnValve}: rotating a valve handle around a revolute joint.
\end{itemize}

\begin{figure}
\begin{subfigure}{\linewidth}
    \centering
    \captionsetup{font=normalsize}
    \includegraphics[width=0.32\linewidth, height=120pt, trim=0 0 60 0, clip]{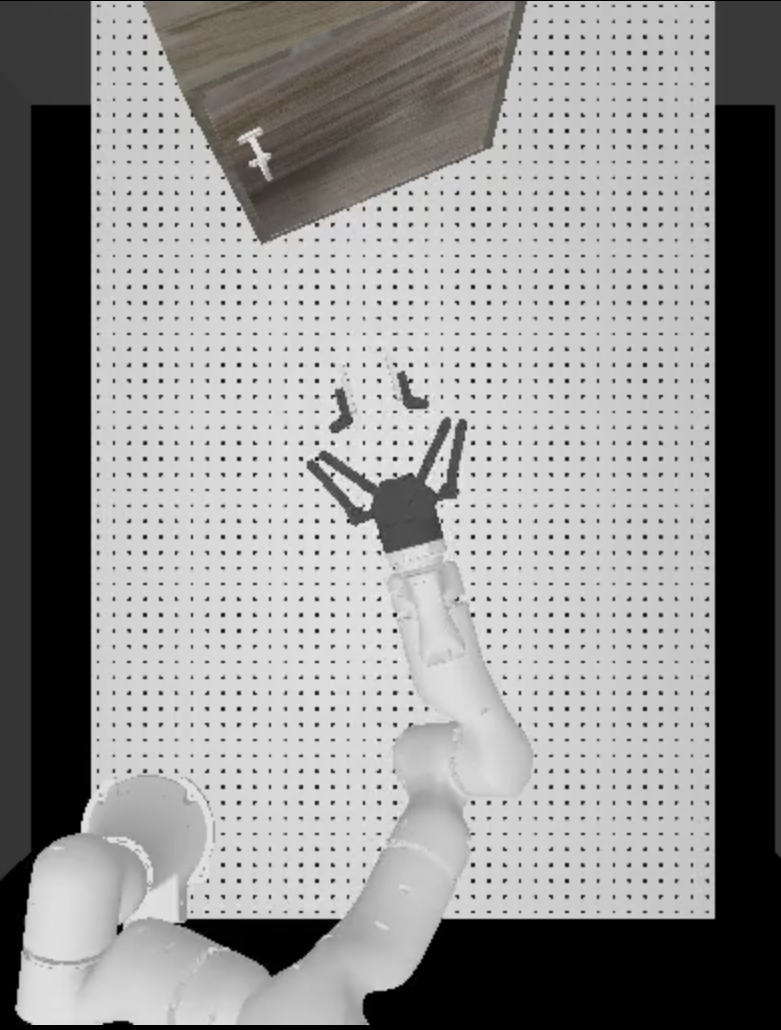}
    \hfill
    \includegraphics[width=0.32\linewidth, height=120pt, trim=0 0 60 0, clip]{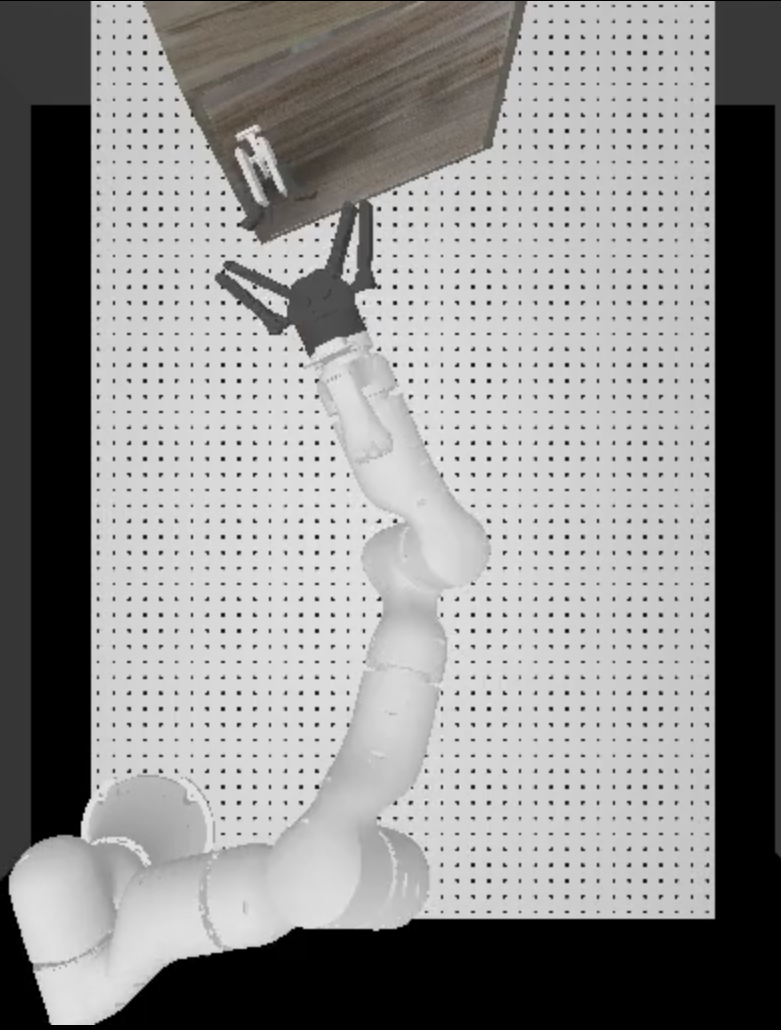}
    \hfill
    \includegraphics[width=0.32\linewidth, height=120pt, trim=20 0 0 0, clip]{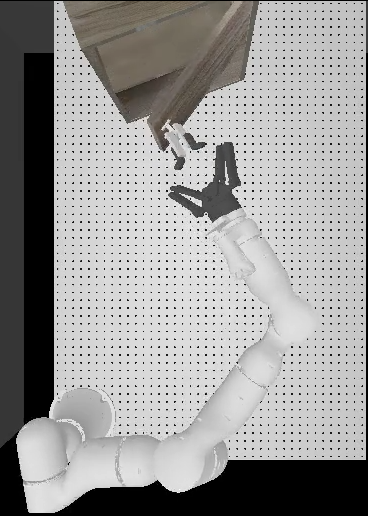}
    \hfill
    \caption{OpenDoor task.}
    \label{fig:suba}
    \hfill
\end{subfigure}
\hfill
\begin{subfigure}{\linewidth}
    \centering
    \captionsetup{font=normalsize}
    \includegraphics[width=0.32\linewidth, height=120pt, trim=0 0 50 0, clip]{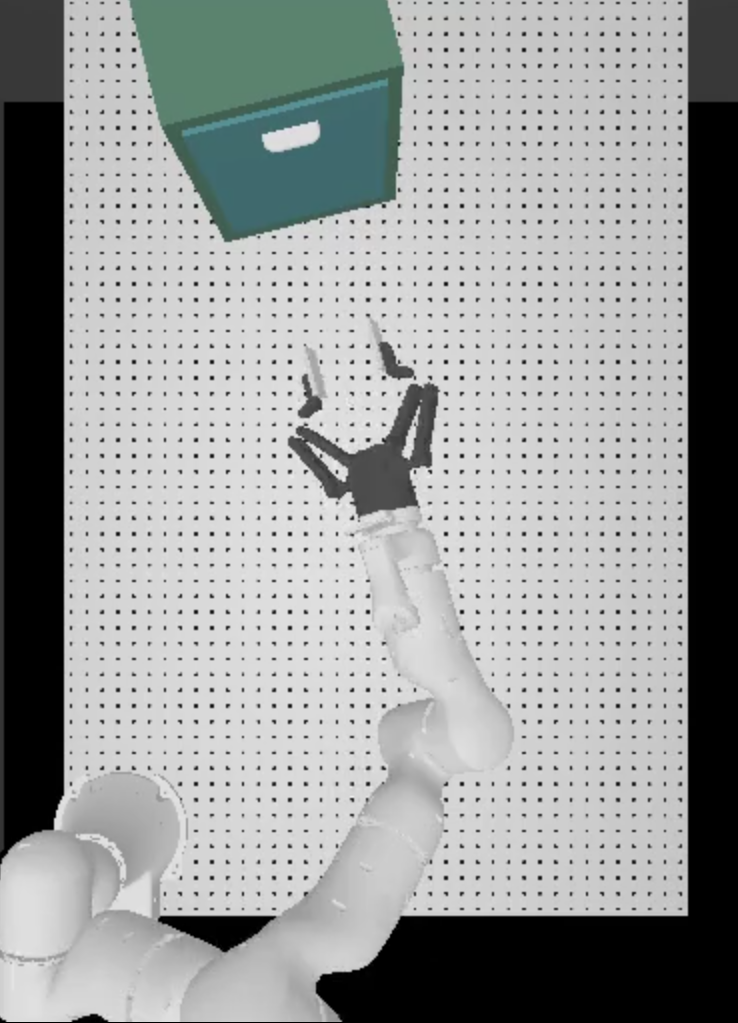}
    \hfill
    \includegraphics[width=0.32\linewidth, height=120pt, trim=0 0 50 0, clip]{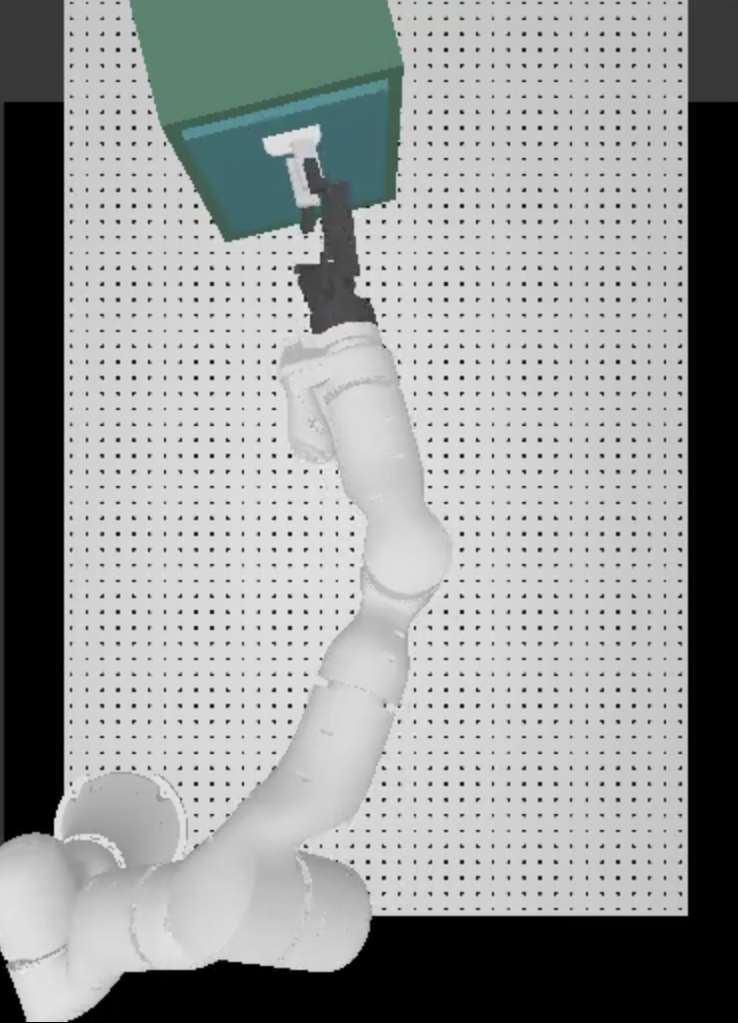}
    \hfill
    \includegraphics[width=0.32\linewidth, height=120pt, trim=20 0 0 0, clip]{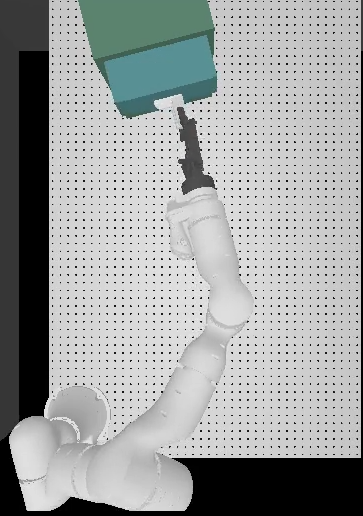}
    \caption{OpenDrawer task.}
    \label{fig:subb}
    \hfill
\end{subfigure}
\hfill
\begin{subfigure}{\linewidth}
    \centering
    \captionsetup{font=normalsize}
    \includegraphics[width=0.32\linewidth, height=120pt]{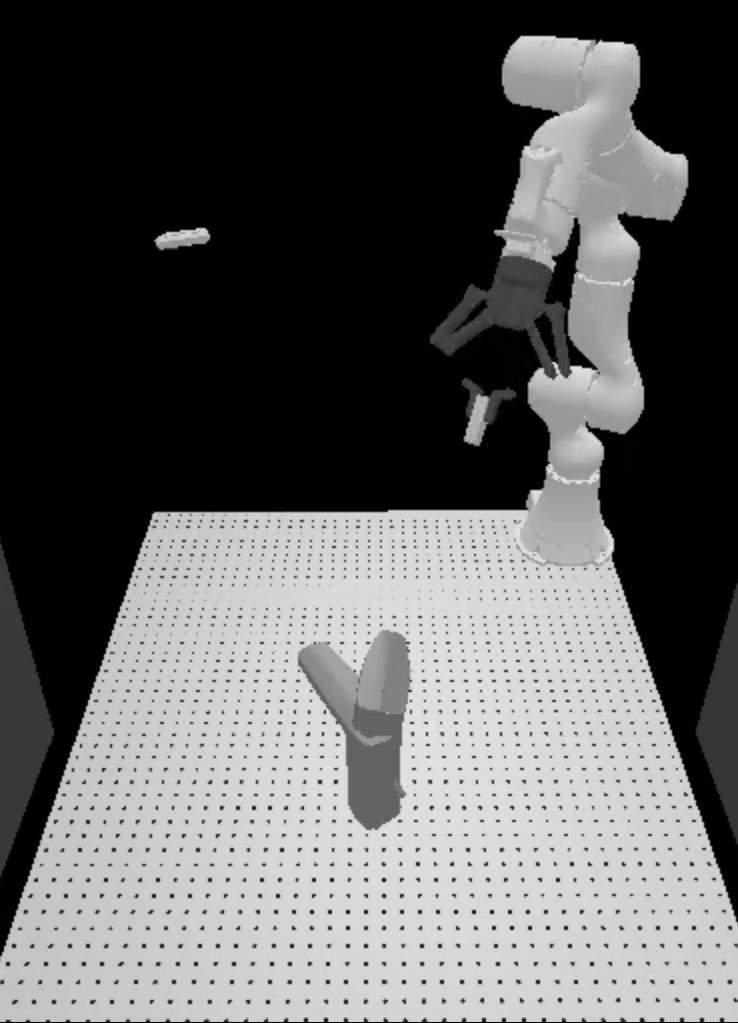}
    \hfill
    \includegraphics[width=0.32\linewidth, height=120pt]{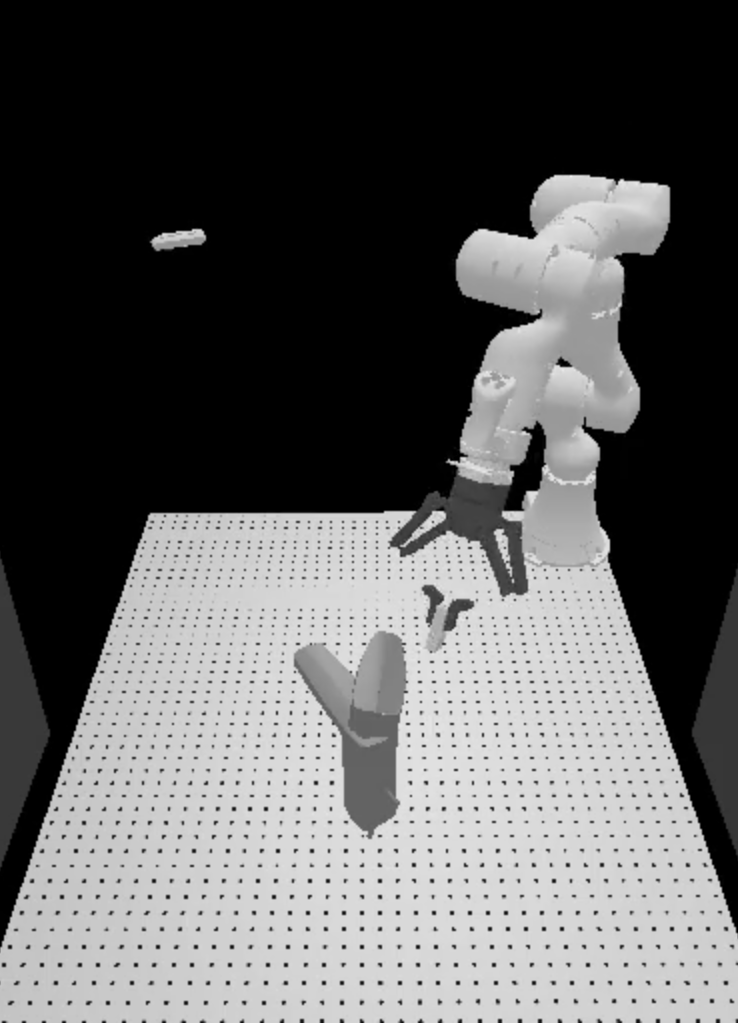}
    \hfill
    \includegraphics[width=0.32\linewidth, height=120pt]{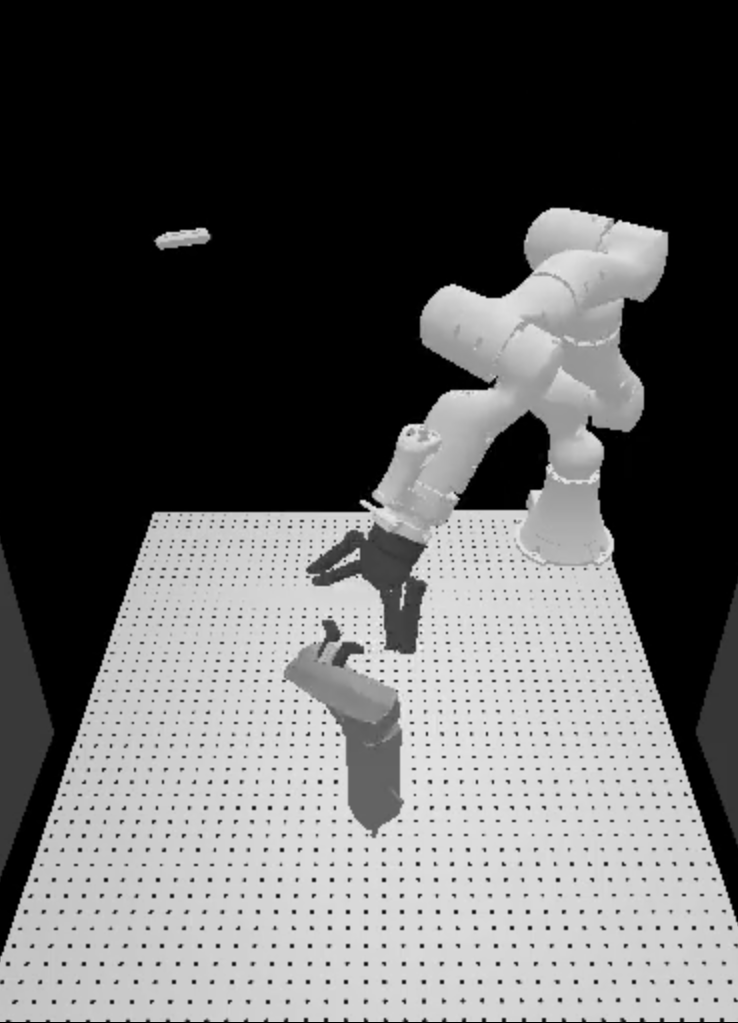}
    \caption{TurnValve task.}
    \label{fig:subc}
\end{subfigure}
\caption{{\normalsize{Tasks in simulated environments.}}}

\label{fig:env_repre}
\end{figure}

These tasks span distinct articulation types and interaction affordances, and together cover a broad set of real-world maintenance behaviours.

\textbf{Baselines}
To assess the benefit of energy-aware constrained optimisation, we compare the following methods:
\begin{itemize}
    \item \textbf{SAC.} A standard SAC agent without cost modelling or constraint enforcement. This baseline represents conventional RL approaches commonly used in robotic manipulation.
    \item \textbf{Constrained SAC (Ours).} The full framework described in Section~\ref{overall}, which integrates part-guided perception, articulation-agnostic control, and a Lagrangian-based constrained SAC.
    
\end{itemize}
Both agents share identical network architectures, perception modules, and training settings to ensure fair comparison.

\begin{figure}
\begin{subfigure}{\linewidth}
    \centering
    \captionsetup{font=normalsize}
    \includegraphics[width=\linewidth, height=75pt]{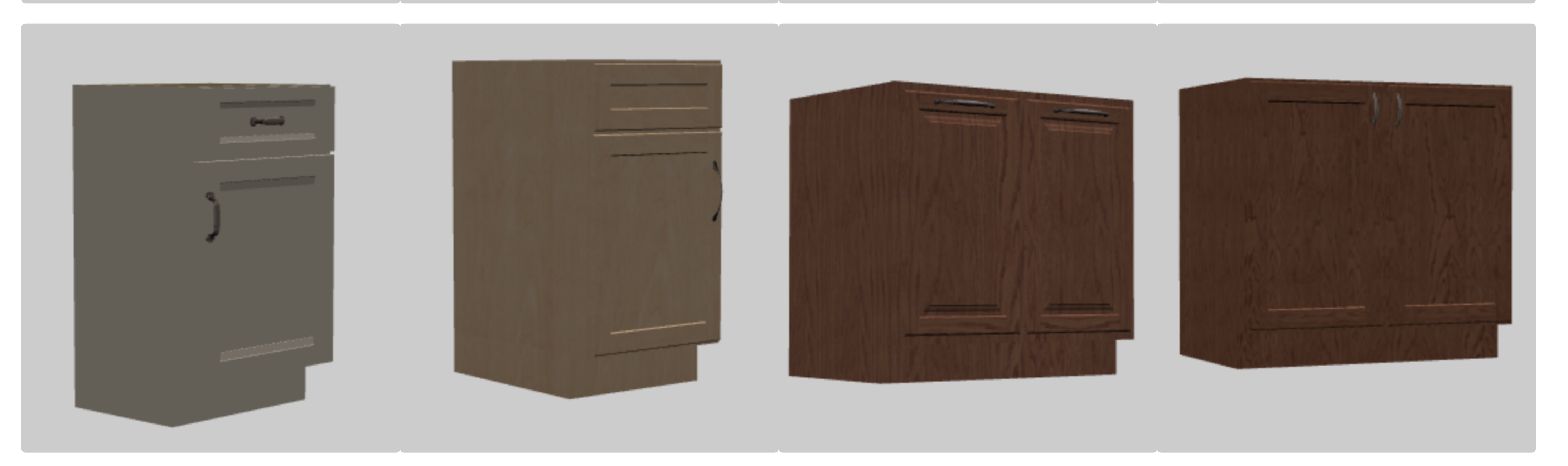}
    \caption{Cabinets with doors.}
    \label{fig:suba}
\end{subfigure}
\hfill
\begin{subfigure}{\linewidth}
    \centering
    \captionsetup{font=normalsize}
    \includegraphics[width=\linewidth, height=75pt]{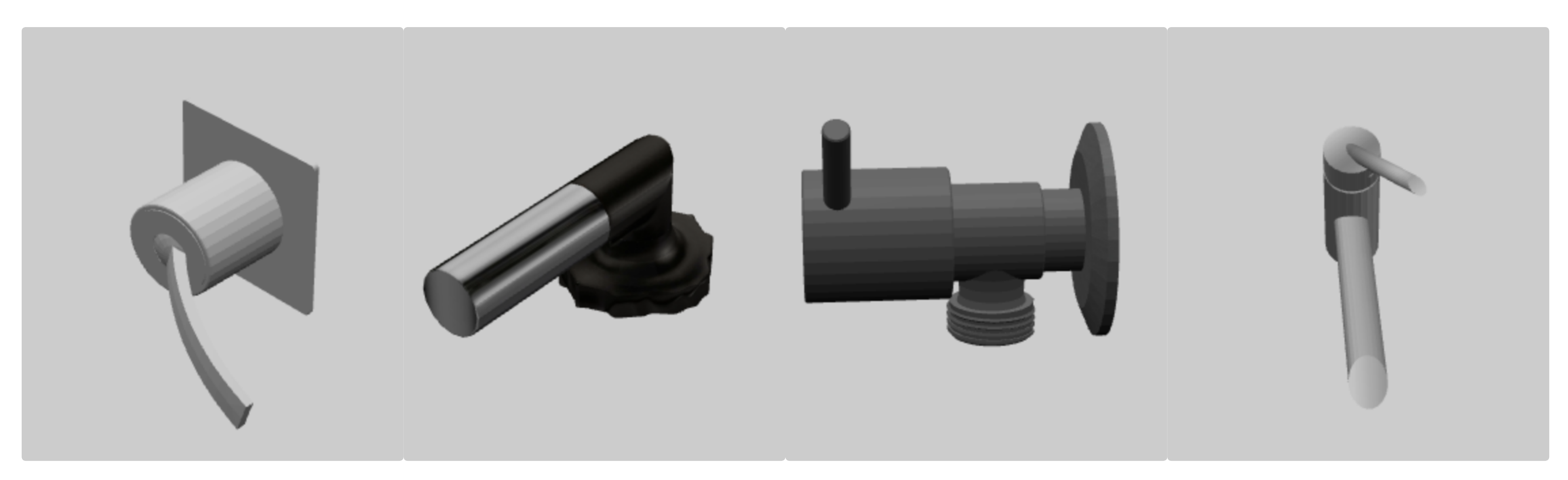}
    \caption{Chests of drawers.}
    \label{fig:subb}
\end{subfigure}
\hfill
\begin{subfigure}{\linewidth}
    \centering
    \captionsetup{font=normalsize}
    \includegraphics[width=\linewidth, height=75pt]{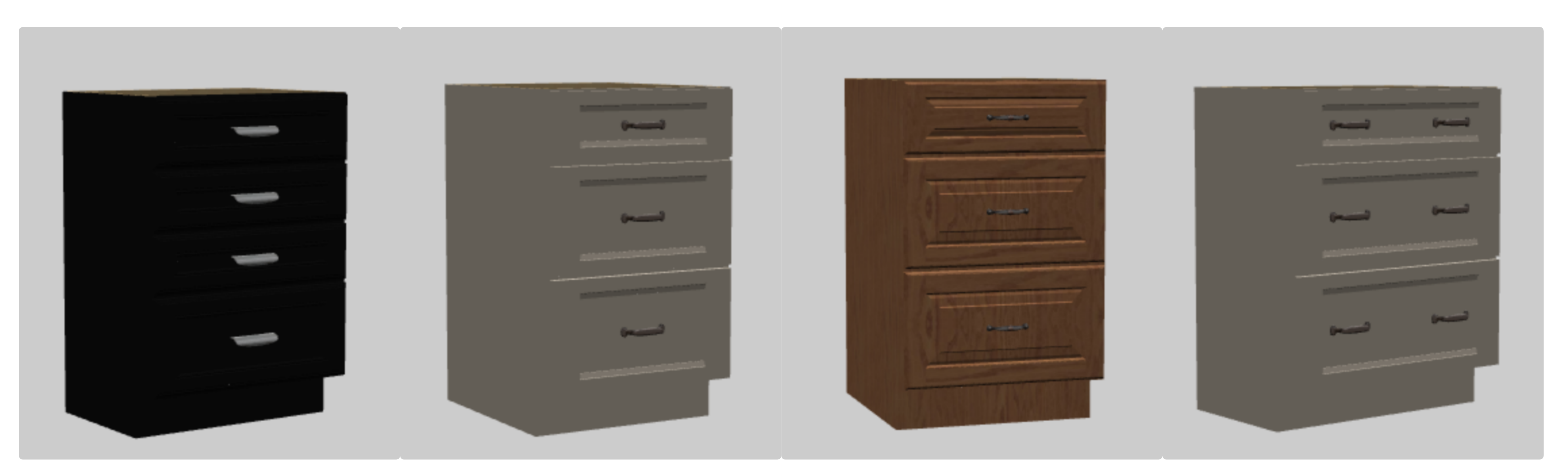}
    \caption{Valves.}
    \label{fig:subc}
\end{subfigure}

\caption{Some of the objects from 3 categories used in training and evaluation. All objects used in evaluation are not used in training, which are unseen to the pretrained models.}
\label{fig:dataset}
\end{figure}


\subsection{Energy Consumption Comparison with SAC}
\label{sec:energy_results}

We first examine whether the proposed energy-aware control scheme improves actuation efficiency compared with standard SAC. 
Figure~\ref{fig:energy_curve} provides a training-time view of total energy consumption under both methods, while Table~\ref{tab:energy} reports quantitative results after convergence. {
In addition to the unconstrained SAC baseline (SAC without $r^{Energy}$), we further compare our method with a commonly used energy-aware variant in which an energy penalty term is directly incorporated into the reward function, denoted as \emph{SAC with $r^{Energy}$}. This baseline represents a standard reward-shaping approach to energy regulation.}

From Table~\ref{tab:energy}, compared with SAC without $r^{Energy}$, Constrained SAC consistently reduces the total energy cost required to complete each task: on OpenDoor, the total energy drops from $79.34 \pm 3.10$ to $63.86 \pm 1.35$, corresponding to a \textbf{19.51\%} saving;
on OpenDrawer, energy decreases from $78.11 \pm 1.18$ to $65.24 \pm 4.04$ (\textbf{16.49\%} saving);
and on TurnValve, energy is reduced from $52.80 \pm 3.77$ to $36.91 \pm 1.49$, yielding a \textbf{30.08\%} improvement.
The bar plots in Figure~\ref{fig:energy_curve} further show that these gaps emerge early during training and remain stable as learning progresses. 
{
\color{black}
In addition, when compared with the reward-penalty baseline SAC with $r^{Energy}$, Constrained SAC still demonstrates clear advantages in energy efficiency.
Specifically, on OpenDoor, the total energy is further reduced from $64.71 \pm 2.38$ to $63.86 \pm 1.35$, corresponding to an additional \textbf{1.31\%} saving;
on OpenDrawer, energy consumption decreases from $75.16 \pm 1.80$ to $65.24 \pm 4.04$, yielding a \textbf{13.19\%} reduction;
and on TurnValve, energy drops from $49.75 \pm 2.94$ to $36.91 \pm 1.49$, resulting in a substantial \textbf{25.81\%} improvement.
These results indicate that explicitly enforcing energy as a constraint achieves consistently lower energy consumption than simply penalising energy in the reward.}

\begin{table}
\centering
\caption{Comparison of total energy consumption on different tasks.
\label{tab:energy}}
\begin{tabular}{@{}lccc@{}}
\rowcolor{headergray}
\textbf{Task} & \textbf{OpenDoor} & \textbf{OpenDrawer} & \textbf{TurnValve} \\
Constrained SAC (Ours) & 63.86 $\pm$ 1.35 & 65.24 $\pm$ 4.04 & 36.91 $\pm$ 1.49 \\
\rowcolor{rowgray}
{SAC without $r^{Energy}$}  & 79.34 $\pm$ 3.10 & 78.11 $\pm$ 1.18 & 52.80 $\pm$ 3.77 \\
Our Energy Savings & 19.51\% & 16.49\% & 30.08\% \\
\rowcolor{rowgray}
{SAC with $r^{Energy}$} & {64.71 $\pm$ 2.38} & {75.16 $\pm$ 1.80 }& {49.75 $\pm$ 2.94 }\\
{Our Energy Savings} & {1.31\%} & {13.19\%} & {25.81\%} \\
\end{tabular}
\begin{tablenotes}
\item Notes: Results are the average of 5 sets of experiments with unique seeds, with each set attempted 50 times.
\end{tablenotes}
\end{table}

\begin{figure*}[t]
\centering

\centerline{\includegraphics[width=\linewidth,height=16pc]{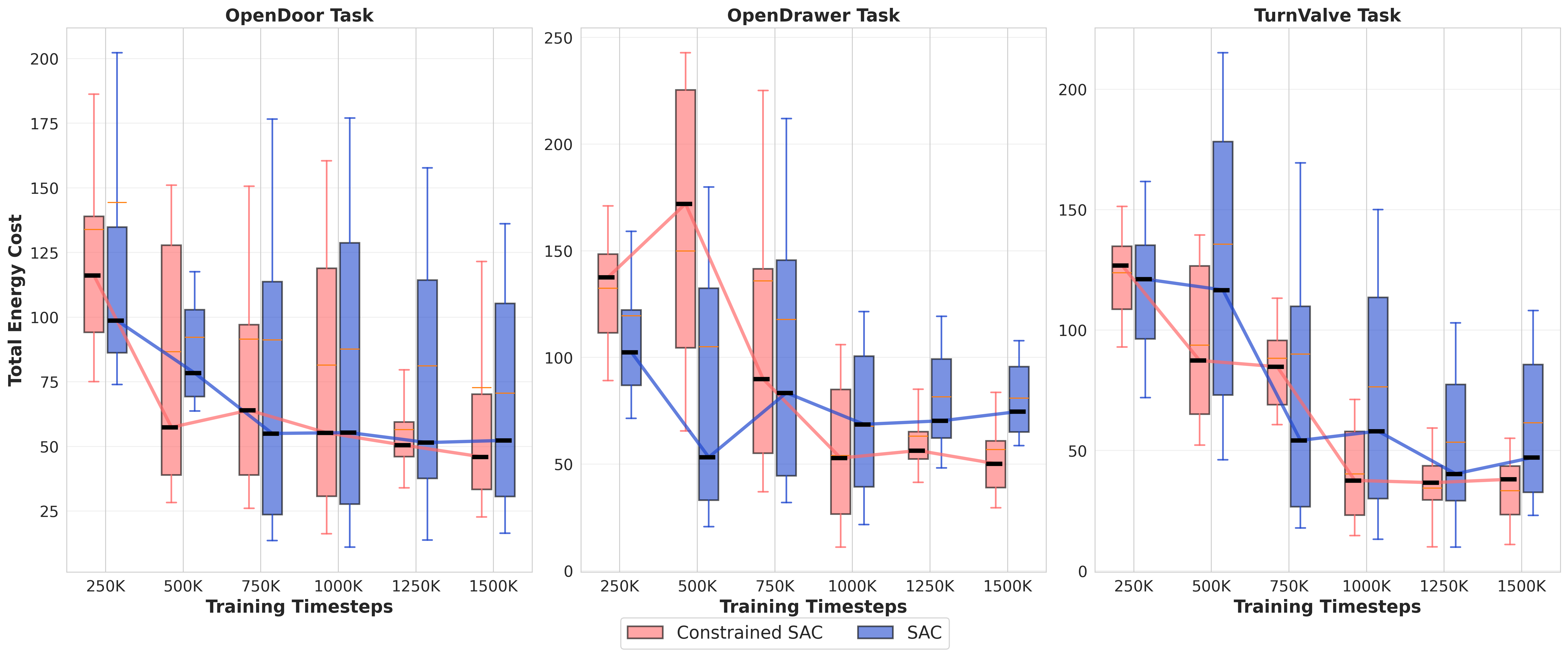}}
\caption{Total energy cost during training for the three tasks. For each checkpoint, the central box represents the interquartile range (IQR), spanning from the first quartile (Q1) to the third quartile (Q3), the horizontal line inside the box denotes the mean total energy cost per episode for Constrained SAC and SAC, and the whiskers extend from the edges of the box to the smallest and largest values from Q1 and Q3, respectively; the overlaid lines indicate the temporal evolution of the corresponding averages.}
\label{fig:energy_curve}
\end{figure*}

These results provide a clear positive answer to \textbf{RQ1}: explicitly modelling energy as a CMDP cost and optimising with Constrained SAC leads to significantly more energy-efficient manipulation than an unconstrained SAC baseline.


\subsection{Manipulation Efficiency: Steps to Completion}
\label{sec:steps_results}

We next study how the energy-aware constraint affects manipulation efficiency, measured by the number of control steps required to complete a task.
Table~\ref{tab:steps} summarises the results.

\begin{table}
\centering
\caption{Comparison of total steps to success on different tasks.
\label{tab:steps}}
\begin{tabular}{@{}lccc@{}}
\rowcolor{headergray}
\textbf{Task} & \textbf{OpenDoor} & \textbf{OpenDrawer} & \textbf{TurnValve} \\
Constrained SAC & 28.28 $\pm$ 0.62 & 31.63 $\pm$ 2.49 & 23.96 $\pm$ 1.20 \\
\rowcolor{rowgray}
{SAC without $r^{Energy}$ }   & 41.69 $\pm$ 2.57 & 37.98 $\pm$ 0.83 & 29.51 $\pm$ 1.81 \\
Our Step Reduction   & 32.16\% & 16.72\% & 18.81\% \\
\rowcolor{rowgray}
{SAC with $r^{Energy}$ } & {67.15 $\pm$ 5.76} & {29.91 $\pm$ 2.15 }& {29.48 $\pm$ 1.54 }\\
{Our Step Reduction}   & {57.89\%} & {-5.75\% }& {18.72\% }\\
\end{tabular}
\begin{tablenotes}
\item Notes: Results are the average of 5 sets of experiments with unique seeds, with each set attempted 50 times.
\end{tablenotes}
\end{table}

Constrained SAC achieves shorter trajectories on all three tasks when compared with the SAC without $r^{Energy}$. 
For OpenDoor, the required steps decrease from $41.69 \pm 2.57$ (SAC) to $28.28 \pm 0.62$, a \textbf{32.16\%} reduction. 
On OpenDrawer, steps are reduced from $37.98 \pm 0.83$ to $31.63 \pm 2.49$ (\textbf{16.72\%} reduction), 
and on TurnValve from $29.51 \pm 1.81$ to $23.96 \pm 1.20$ (\textbf{18.81\%} reduction).{
\color{black}
In addition, when compared with the reward-penalty baseline SAC with $r^{Energy}$, Constrained SAC still demonstrates clear advantages in manipulation efficiency.
Specifically, on OpenDoor, the required number of steps is reduced from $67.15 \pm 5.76$ to $28.28 \pm 0.62$, yielding a substantial \textbf{57.89\%} reduction.
On OpenDrawer, SAC with $r^{Energy}$ requires $29.91 \pm 2.15$ steps, which is comparable to Constrained SAC, resulting in a marginal difference of \textbf{-5.75\%}.
On TurnValve, the number of steps decreases from $29.48 \pm 1.54$ to $23.96 \pm 1.20$, corresponding to a \textbf{18.72\%} reduction. Qualitatively, the constrained policy exhibits smoother and more decisive motions, avoiding unnecessary back-and-forth corrections that are frequently observed in both the SAC without $r^{Energy}$ and the reward-penalty baseline.
}

These findings answer \textbf{RQ2}: enforcing an energy constraint not only reduces actuation cost, but also leads to more efficient and direct manipulation trajectories. 
This behaviour is consistent with the first contribution, where articulation-agnostic control is achieved through part-guided perception while remaining compatible with energy-aware optimisation.


\subsection{Task Success Rate and Reliability}
\label{sec:success_results}

Finally, we evaluate how energy-aware constrained optimisation affects the overall reliability of articulated-object manipulation.
Table~\ref{tab:success} reports the success rate on each task.

\begin{table}
\centering
\caption{Comparison of success rate on different tasks.
\label{tab:success}}
\begin{tabular}{@{}lccc@{}}
\rowcolor{headergray}
\textbf{Task} & \textbf{OpenDoor} & \textbf{OpenDrawer} & \textbf{TurnValve} \\
Constrained SAC (\%) & 98.8 $\pm$ 1.79 & 94.0 $\pm$ 3.16 & 56.8 $\pm$ 8.32 \\
\rowcolor{rowgray}
{SAC without $r^{Energy}$ }(\%)  & 83.6 $\pm$ 4.56 & 97.2 $\pm$ 1.79 & 54.8 $\pm$ 8.32  \\
Our Improvement  & 15.2\% & -3.2\% & 2.0\% \\
{SAC with $r^{Energy}$ (\%)} & {2.55 $\pm$ 0.87} & {94.0 $\pm$ 1.53 }& {55.8 $\pm$ 10.36 }\\
{Our Improvement}  &{3774.51\%} & {0.0\%} & { 1.79\%} \\
\end{tabular}
\begin{tablenotes}
\item Notes: Results are the average of 5 sets of experiments with unique seeds, with each set attempted 50 times.
\end{tablenotes}
\end{table}

{\color{black}
Constrained SAC maintains high success rates across all three tasks when compared with the SAC without $r^{Energy}$.} On OpenDoor, Constrained SAC improves the success rate from $83.6 \pm 4.56\%$ to $98.8 \pm 1.79\%$, a \textbf{15.2\%} absolute gain.
On OpenDrawer, the success rates of both methods are comparable (Constrained SAC: $94.0 \pm 3.16\%$, SAC: $97.2 \pm 1.79\%$), 
indicating that the energy constraint does not hinder performance on this relatively easier task.
On TurnValve, both methods achieve similar success rates (within the reported variance), with a slight advantage for Constrained SAC. {\color{black}
In addition, when compared with the reward-penalty baseline SAC with $r^{Energy}$, the advantages of the constrained formulation become more pronounced.
Specifically, on OpenDoor, the success rate increases dramatically from $2.55 \pm 0.87\%$ to $98.8 \pm 1.79\%$, representing a \textbf{3774.51\%} relative improvement.
On OpenDrawer, both methods achieve identical success rates ($94.0\%$), resulting in no measurable difference.
On TurnValve, the success rate improves slightly from $55.8 \pm 10.36\%$ to $56.8 \pm 8.32\%$, corresponding to a \textbf{1.79\%} increase.
}

{Compared with OpenDoor and OpenDrawer, the TurnValve task exhibits a lower success rate, reflecting the increased difficulty of manipulation under higher mechanical resistance and tighter rotational constraints. Empirically, unsuccessful episodes are often accompanied by insufficient torque buildup in early interaction stages or unstable contact near the rotational joint, limiting sustained rotation.}

Across all three scenarios, the constrained policy maintains high success rates while achieving substantial energy savings and step reductions. 
This demonstrates that the proposed framework does not trade task performance for energy efficiency; instead, it achieves a favourable balance between the two, in line with the multi-objective reward design in Section~\ref{sec:mdp_details}. 
These observations provide empirical support for \textbf{RQ3}: the end-to-end constrained optimisation yields reliable performance suitable for long-term infrastructure operation, rather than merely optimising a single-task metric.

\subsection{Constraint Satisfaction and Training Dynamics}
\label{sec:convergence}

Figure~\ref{fig:constraint_cost} plots the evolution of the constraint cost during training for the three articulated-object tasks, together with the corresponding violation threshold~$\theta$. {The violation threshold is not fixed throughout training; instead, \textcolor{black}{it is scheduled in a piecewise-constant manner. Specifically, it is initialised at $0.060$, reduced to $0.055$ and $0.050$ at $3\times10^5$ and $6\times10^5$ training steps, respectively, and finally tightened to $0.045$ after $10^6$ steps, as illustrated by the red dashed curve in Fig.~\ref{fig:constraint_cost}.}
This design allows the agent to first acquire basic manipulation competence under a relatively loose energy budget, before gradually enforcing stricter energy regulation.}

Across OpenDoor, OpenDrawer, and TurnValve, the constrained SAC agent initially explores higher-energy behaviours, after which the dual-ascent update progressively tightens the constraint and suppresses excessive actuation.
As training proceeds, the average constraint cost of all tasks converges to values at or below the threshold, confirming that the long-horizon energy constraint is effectively enforced in practice.

These trends directly support \textbf{RQ4}: the proposed Lagrangian constrained formulation yields stable primal--dual dynamics and achieves consistent constraint satisfaction over the course of training, in line with the theoretical properties discussed in Section~\ref{sec:optimization}.

\begin{figure}[t]
\centering
\centerline{\includegraphics[width=\linewidth,height=16pc]{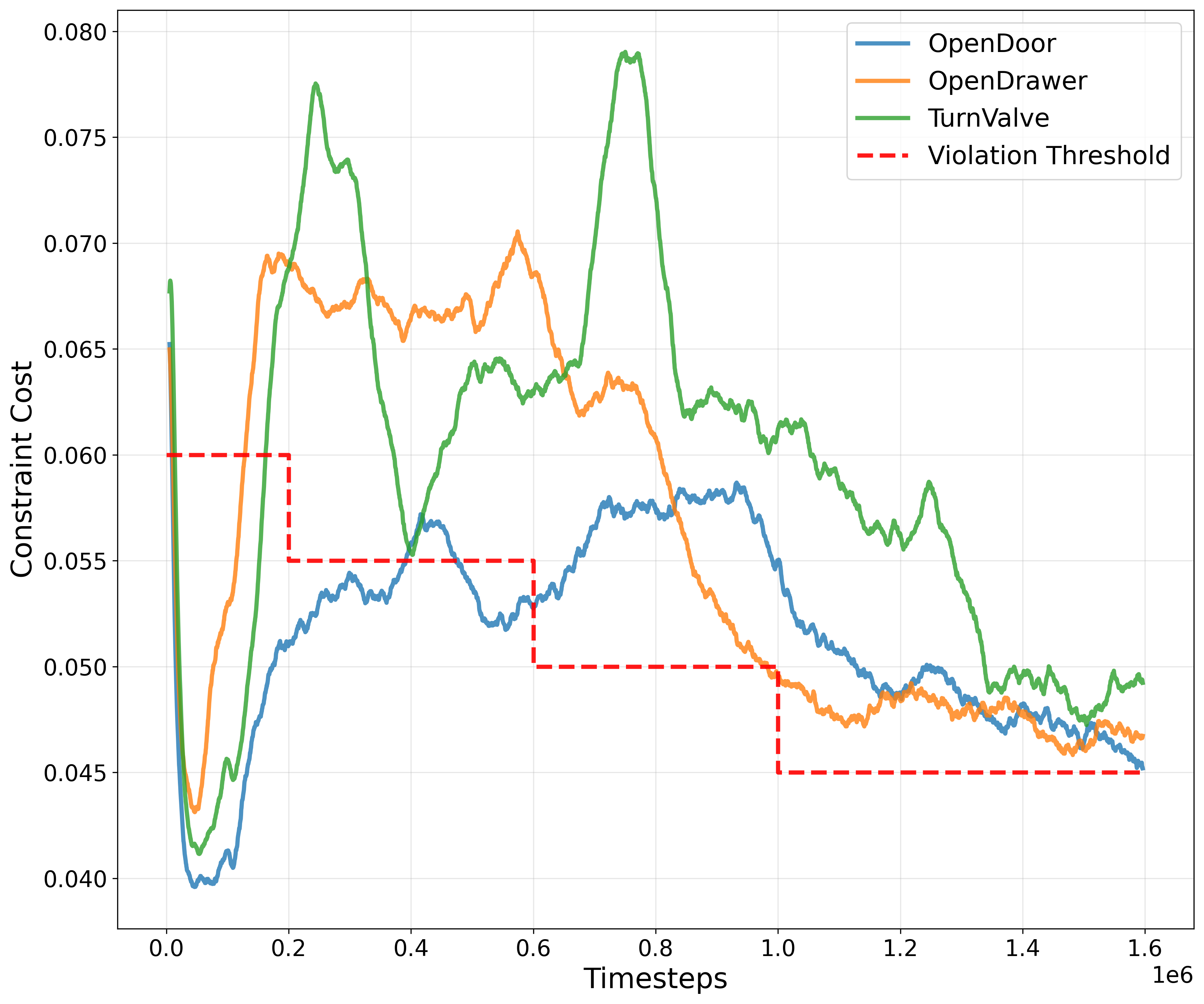}}
\caption{Training dynamics of the constraint cost for the three tasks. Solid lines denote the mean constraint cost across seeds; the red dashed curve indicates the violation threshold~$\theta$ scheduled during training.}
\label{fig:constraint_cost}
\end{figure}

Taken together, the results in Sections~\ref{sec:energy_results}-\ref{sec:convergence} show that the proposed articulation-agnostic, energy-aware manipulation framework  
(i) substantially reduces energy consumption, 
(ii)  shortens manipulation trajectories, and  
(iii) maintains or improves task success rates across diverse articulated infrastructure components,
(iv) satisfies the prescribed energy constraint during training.

{\color{black}
\subsection{Perception Module Reliability and Runtime Analysis}
\label{sec:perception_results}

Although perception accuracy is not the primary objective of this work, we report segmentation performance and inference latency to verify that the part-guided perception pipeline provides stable and real-time inputs for downstream control.

\begin{table*}
\centering
\caption{Segmentation Performance and Inference Time Statistics}
\label{tab:seg}
\begin{tabular}{@{}lccccccc@{}}
\rowcolor{headergray}
{\textbf{Metric}} & {\textbf{Overall}} & {\textbf{Handle}} & {\textbf{Door}} & {\textbf{Cabinet}} & {\textbf{FixLink}} & {\textbf{SwitchLink}} & {\textbf{Other}}\\
{mIoU} & {0.4372} & {0.1945} & {0.3614} & {0.1404} & {0.4325} & {0.3879} & {0.6218}\\
\rowcolor{rowgray}
{mF1} & {0.6084} & {0.2598} & {0.4529} & {0.2062} & {0.4405} & {0.4009} & {0.7435} \\
{Average Time} & \multicolumn{7}{c}{{2.34 $\pm$ 0.12 ms}} \\
\rowcolor{rowgray}
{Min/Max Time} & \multicolumn{7}{c}{{2.24 / 2.82 ms}} \\
{Throughput} & \multicolumn{7}{c}{{427.27 FPS}} \\
\end{tabular}
\begin{tablenotes}
\item {Notes: mIoU (mean Intersection over Union) and mF1 (mean F1 score) are computed using multiclass segmentation metrics with micro reduction across six classes (handle, door, cabinet, fixlink, switchlink, other). Inference time statistics are measured on a real image dataset with GPU synchronization.}
\end{tablenotes}
\end{table*}

As shown in Table~\ref{tab:seg}, the segmentation network achieves an overall mIoU of 0.437 and an mF1 score of 0.608 across all part categories. Performance varies across different parts: higher mIoU and mF1 scores are observed on visually and geometrically stable regions (e.g., \emph{FixLink} and \emph{Other}), whereas smaller or less visually distinctive parts (e.g., \emph{Handle}) exhibit lower accuracy. This behaviour is expected given the limited pixel footprint and increased visual ambiguity of small functional parts in RGB observations.

In terms of computational efficiency, the perception pipeline achieves an average inference time of $2.34 \pm 0.12$\,ms per frame, corresponding to a throughput exceeding 400 FPS. \textcolor{black}{The reported inference latency in Table~\ref{tab:seg} corresponds to the full perception-to-state pipeline, including part segmentation, uncertainty estimation, weighted point sampling, and PointNet-based geometric encoding.}

Overall, the results in Table~\ref{tab:seg} indicate that the proposed part-guided perception module provides sufficiently reliable and computationally efficient inputs for energy-aware manipulation, without becoming a bottleneck for real-time infrastructure O\&M scenarios.
}

{\color{black}
\section{Conclusions}\label{sec6}

This paper presents an articulation-agnostic and energy-aware robotic manipulation framework for intelligent infrastructure O\&M. By integrating part-guided 3D geometric perception with a constrained RL controller, the proposed method could generalize across heterogeneous articulated components while explicitly regulating long-horizon energy consumption.

Extensive simulation experiments on three representative O\&M tasks—door opening, drawer pulling, and valve turning—demonstrate the effectiveness of the proposed approach. Quantitatively, the constrained SAC controller achieves consistent energy savings of approximately 16--30\% compared with the unconstrained SAC baseline, while simultaneously reducing the number of execution steps by 16--32\% without sacrificing task success rates. These results indicate that explicitly enforcing energy as a constraint leads to smoother trajectories, fewer redundant motions, and more decisive manipulation behavior.

Moreover, the proposed framework exhibits robust generalization across different articulated mechanisms under a unified formulation. The combination of part-guided segmentation, weighted point sampling, and lightweight PointNet-based encoding enables stable performance across revolute and prismatic joints, validating the scalability of the articulation-agnostic perception--control pipeline. 

Overall, the results demonstrate that energy-aware constrained reinforcement learning, when coupled with function-centric geometric perception, provides a practical and scalable solution for long-term robotic manipulation in infrastructure O\&M settings, balancing task effectiveness and operational efficiency.

{Future work will extend the proposed framework to real-world robotic platforms and systematically study sim-to-real transfer under realistic infrastructure inspection conditions, such as sensor noise, illumination variation, and surface contamination, to further assess robustness in long-term deployment. In addition, extending the framework to more complex multi-joint articulated manipulation scenarios remains an important direction for future work.}
{
Beyond the evaluated O\&M scenarios, the proposed framework is applicable to other engineering domains such as industrial equipment operation, facility management, and energy-aware robotic manipulation, where articulated interaction and long-term efficiency are critical.
}{The current cost does not explicitly model actuator torque or electrical power. Incorporating torque- or power-based energy models, as well as conducting sensitivity analyses under alternative cost definitions, is left for future work.
}

{The TurnValve task reveals the increased difficulty of manipulation under high mechanical resistance. Addressing such cases through torque-aware control and improved joint-centric perception constitutes an important direction for future work. }{While the current study focuses on generalization across articulated geometries and interaction configurations, future work will incorporate physics-level domain randomization (e.g., friction and damping variation) to further improve robustness for sim-to-real deployment.} {Future work will explore neural dynamic classification to improve robustness under non-stationary conditions \cite{rafiei2017new}. Fast-learning supervised models such as the Finite Element Machine provide a promising direction for improving learning efficiency \cite{pereira2020fema}. Dynamic ensemble learning may further enhance adaptability in long-term deployment \cite{alam2020dynamic}. Self-supervised learning will be investigated to improve representation robustness without manual annotation \cite{rafiei2022self}.}
}

\textcolor{black}{Exploring adaptive fusion strategies for uncertainty-aware and temporal-consistency cues, and systematically evaluating their impact on representation stability and downstream policy learning, remains an important direction for future work.} \textcolor{black}{Future work will further consider more realistic joint dynamics, such as increased friction caused by mechanical wear, and evaluate their effects on task success, energy consumption, and the number of manipulation steps.}

\bmsection*{Acknowledgments}
This research was supported by the Joint Funds of the National Natural Science Foundation of China under Grant U1864206. Xiaowen Tao and Yinuo Wang contribute equally to this work. Note: This version supersedes all previous preprint versions. For reuse of any content in any version of this work, please contact the corresponding author. Upon journal publication, copyright may be transferred to the publisher, and reuse of content from the published version should follow the publisher's permissions process.

{\color{black}
\bibliography{wileyNJD-APA}
}

\end{document}